%% file: esop2016.tex
\documentclass{llncs}
\usepackage[utf8]{inputenc}
\usepackage{makeidx}
\usepackage{enumerate}
\usepackage{hyperref}
\usepackage{latexsym, amsmath, amssymb}
\usepackage{stmaryrd} 
\usepackage{bussproofs}
\usepackage{cite}
\usepackage{bbding}
\usepackage{enumitem}
\usepackage{listings}
\usepackage{tikz}
\usetikzlibrary{matrix}
\usepackage{url}

\newcommand\dbot{\bot\!\!\!\bot}
\newcommand\converge{\Downarrow}
\newcommand\diverge{\Uparrow}
\newcommand\red{\twoheadrightarrow}
\newcommand\reds{\twoheadrightarrow^{*}}
\newcommand\bbN{\mathbb{N}}
\newcommand\vvdash{\vdash_{\!\!\!\text{val}}}
\newcommand\scissors{\texttt{8<}}
\newcommand{\bnfeq}{\ensuremath{\;::=\;}}
\newcommand{\bnfor}{\ensuremath{\;\;\vert\;\;}}
\newcommand{\tcase}[2]{\ensuremath{\text{case}_{#1}\;[#2]}}

\newcommand{\vsem}[1]{\ensuremath{\llbracket #1 \rrbracket}}
\newcommand{\ssem}[1]{\ensuremath{\llbracket #1 \rrbracket^{\bot}}}
\newcommand{\tsem}[1]{\ensuremath{\llbracket #1 \rrbracket^{\bot\bot}}}

\newenvironment{mproof}{\leavevmode}{\DisplayProof}

\begin{document}
\mainmatter

\title{A Classical Realizability Model for a\\ Semantical Value Restriction}
\titlerunning{Semantical Value Restriction}
\author{Rodolphe Lepigre}
\authorrunning{R. Lepigre}

\institute{
  LAMA, UMR 5127 - CNRS\\
  Université Savoie Mont Blanc, France\\
  \email{rodolphe.lepigre@univ-smb.fr}
}

\maketitle

\begin{abstract}
  \input{abstract.tex}
\end{abstract}

\section*{Introduction}
\input{introduction.tex}

\section{Syntax, Reduction and Equivalence}\label{sect-syntax}
\input{syntax.tex}

\section{Formulas and Semantics}\label{sect-semantics}
\input{semantics.tex}

\section{Deciding Program Equivalence}
\input{equiv.tex}

\section{Further Work}
\input{conclusion.tex}

\section*{Acknowledgments}
\input{acknowledgment.tex}

\bibliographystyle{splncs03}
\bibliography{biblio}

\end{document}

%% file: abstract.tex
We present a new type system with support for proofs of programs
in a call-by-value language with control operators. The proof mechanism
relies on observational equivalence of (untyped) programs. It appears in
two type constructors, which are used for specifying program properties
and for encoding dependent products.
The main challenge arises from the lack of expressiveness of dependent
products due to the value restriction. To circumvent this limitation we
relax the syntactic restriction and only require equivalence to a value.
The consistency of the system is obtained semantically by constructing a
classical realizability model in three layers (values, stacks and terms).

%% file: introduction.tex
In this work we consider a new type system for a call-by-value language,
with control operators, polymorphism and dependent products. It is intended
to serve as a theoretical basis for a proof assistant focusing on program
proving, in a language similar to OCaml or SML. The proof mechanism relies
on dependent products and equality types $t \equiv u$, where $t$ and $u$ are
(possibly untyped) terms of the language. Equality types are interpreted as
$\top$ if the denoted equivalence holds and as $\bot$ otherwise.

In our system, proofs are written using the same language as programs. For
instance, a pattern-matching corresponds to a case analysis in a proof, and
a recursive call to the use of an induction hypothesis. A proof is first
and foremost a program, hence we may say that we follow the ``program as
proof'' principle, rather than the usual ``proof as program'' principle. In
particular, proofs can be composed as programs and with programs to form
proof tactics.

Programming in our language is similar to programming in any dialect of ML.
For example, we can define the type of unary natural numbers, and the
corresponding addition function.
\begin{lstlisting}
type nat = Z[] | S[nat]
let rec add n m = match n with
  | Z[]   $\to$ m
  | S[nn] $\to$ S[add nn m]
\end{lstlisting}
We can then prove properties of addition such as \texttt{add Z[] n} $\equiv$
\texttt{n} for all \texttt{n} in \texttt{nat}. This property can be expressed
using a dependent product over \texttt{nat} and an equality type.
%
%
\begin{lstlisting}
let addZeroN n:nat : (add Z[] n $\equiv$ n) = $\scissors$
\end{lstlisting}
The term $\scissors$ (to be pronounced ``scissors'') can be introduced
whenever the goal is derivable from the context with equational reasoning.
Our first proof is immediate since we have \texttt{add Z[] n} $\equiv$
\texttt{n} by definition of \texttt{add}.

Let us now show that \texttt{add n Z[]} $\equiv$ \texttt{n} for every
\texttt{n} in \texttt{nat}. Although the statement of this property is
similar to the previous one, its proof is slightly more complex and
requires case analysis and induction.
\begin{lstlisting}
let rec addNZero n:nat : (add n Z[] $\equiv$ n) =
  match n with
  | Z[]   $\to$ $\scissors$
  | S[nn] $\to$ let r = addNZero nn in $\scissors$
\end{lstlisting}
In the \texttt{S[nn]} case, the induction hypothesis (i.e. \texttt{add nn Z[]}
$\equiv$ \texttt{nn}) is obtained by a recursive call. It is then used to
conclude the proof using equational reasoning.
Note that in our system, programs that are considered as proofs need to go
through a termination checker. Indeed, a looping program could be used to
prove anything otherwise. The proofs \texttt{addZeroN} and \texttt{addNZero}
are obviously terminating, and hence valid.

Several difficulties arise when combining call-by-value evaluation,
side-effects, dependent products and equality over programs. Most notably,
the expressiveness of dependent products is weakened by the value
restriction: elimination of dependent product can only happen on
arguments that are syntactic values. In other words, the typing rule
\begin{prooftree}
\AxiomC{$\Gamma \vdash t : \Pi_{a:A}\;B$}
\AxiomC{$\Gamma \vdash u : A$}
\BinaryInfC{$\Gamma \vdash t\; u : B[a := u]$}
\end{prooftree}
\smallskip
cannot be proved safe if $u$ is not a value.
This means, for example, that we cannot derive a proof of
\texttt{add (add Z[] Z[]) Z[]} $\equiv$ \texttt{add Z[] Z[]} by applying
\texttt{addNZero} (which has type $\Pi$\texttt{n:nat} $($\texttt{add n Z[]}
$\equiv$ \texttt{n}$)$) to \texttt{add Z[] Z[]} since it is not a value.
The restriction affects regular programs in a similar way. For instance,
it is possible to define a list concatenation function \texttt{append} with
the following type.
\begin{center}
$\Pi$\texttt{n:nat} $\Pi$\texttt{m:nat List(n)} $\Rightarrow$ \texttt{List(m)}
$\Rightarrow$ \texttt{List(add n m)}
\end{center}
However, the \texttt{append} function cannot be used to implement a function
concatenating three lists. Indeed, this would require being able to provide
append with a non-value natural number argument of the form \texttt{add n m}.

Surprisingly, the equality types and the underlying observational
equivalence relation provide a solution to the lack of expressiveness of
dependent products. The value restriction can be relaxed to obtain the rule
\begin{prooftree}
\AxiomC{$\Gamma, u \equiv v \vdash t : \Pi_{a:A}\;B$}
\AxiomC{$\Gamma, u \equiv v \vdash u : A$}
\BinaryInfC{$\Gamma, u \equiv v \vdash t\; u : B[a := u]$}
\end{prooftree}
\smallskip
which only requires $u$ to be equivalent to some value $v$. The same idea
can be applied to every rule requiring value restriction. The obtained
system is conservative over the one with the syntactic restriction.
Indeed, finding a value equivalent to a term that is already a value can
always be done using the reflexivity of the equivalence relation.

Although the idea seems simple, proving the soundness of the new typing rules
semantically is surprisingly subtle. A model is built using classical
realizability techniques in which the interpretation of a type $A$ is spread
among two sets: a set of values $\llbracket A \rrbracket$ and a set of terms
$\llbracket A \rrbracket^{\bot\bot}$. The former contains all values that
should have type $A$. For example, $\llbracket$\texttt{nat}$\rrbracket$
should contain the values of the form \texttt{S[S[...Z[]...]]}. The set
$\llbracket A \rrbracket^{\bot\bot}$ is the completion of
$\llbracket A \rrbracket$ with all the terms behaving like values of
$\llbracket A \rrbracket$ (in the observational sense).
To show that the relaxation of the value restriction is sound, we need the
values of $\llbracket A \rrbracket^{\bot\bot}$ to also be  in
$\llbracket A \rrbracket$. In other words, the completion operation should
not introduce new values. To obtain this property, we need to
extend the language with a new, non-computable instruction internalizing
equivalence. This new instruction is only used to build the model, and will
not be available to the user (nor will it appear in an implementation).

\subsection*{About effects and value restriction}

A soundness issue related to side-effects and call-by-value evaluation
arose in the seventies with the advent of ML. The problem stems from a
bad interaction between side-effects and Hindley-Milner polymorphism. It was
first formulated in terms of references \cite[section 2]{wright2}, and many
alternative type systems were designed (e.g. \cite{tofte, damas, leroy,
leroycbn}). However, they all introduced a complexity that contrasted with
the elegance and simplicity of ML's type system (for a detailed account, see
\cite[section 2]{wright1} and \cite[section 2]{garigue}).

A simple and elegant solution was finally found by Andrew Wright in the
nineties. He suggested restricting generalization in let-bindings\footnote{In
ML the polymorphism mechanism is strongly linked with let-bindings. In OCaml
syntax, they are expressions of the form \texttt{let x = u in t}.}
to cases where the bound term is a syntactic value \cite{wright1, wright2}.
In slightly more expressive type systems, this restriction appears in the
typing rule for the introduction of the universal quantifier. The usual
rule
\begin{prooftree}
\AxiomC{$\Gamma \vdash t : A$}
\AxiomC{$X \not\in FV(\Gamma)$}
\BinaryInfC{$\Gamma \vdash t : \forall X\;A$}
\end{prooftree}
\smallskip
cannot be proved safe (in a call-by-value system with side-effects) if $t$ is
not a syntactic value.
Similarly, the elimination rule for dependent product (shown previously)
requires value restriction. It is possible to exhibit a counter-example
breaking the type safety of our system if it is omitted \cite{long}.

In this paper, we consider control structures, which have been shown to give
a computational interpretation to classical logic by Timothy Griffin
\cite{griffin}.
In 1991, Robert Harper and Mark Lillibridge found a complex program breaking
the type safety of ML extended with Lisp's \emph{call/cc} \cite{smllist}.
As with references, value restriction solves the inconsistency and yields a
sound type system. 
Instead of using control operators like \emph{call/cc}, we adopt the
syntax of Michel Parigot's $\lambda\mu$-calculus \cite{parigot}. Our language
hence contains a new binder $\mu \alpha\,t$ capturing the continuation in the
$\mu$-variable $\alpha$. The continuation can then be restored in $t$ using
the syntax $u\ast\alpha$\footnote{This was originally denoted $[\alpha]u$.}.
In the context of the $\lambda\mu$-calculus, the soundness issue arises when
evaluating $t\,(\mu\alpha\,u)$ when $\mu\alpha\,u$ has a polymorphic type.
Such a situation cannot happen with value restriction since $\mu\alpha\,u$ is
not a value.

\subsection*{Main results}

The main contribution of this paper is a new approach to value restriction.
The syntactic restriction on terms is replaced by a semantical restriction
expressed in terms of an observational equivalence relation denoted
$(\equiv)$. Although this approach seems simple, building a model to prove
soundness semantically (theorem \ref{adequacy}) is surprisingly subtle.
Subject reduction is not required here, as  our model construction implies
type safety (theorem \ref{safety}). Furthermore our type system is consistent
as a logic (theorem \ref{consistency}).

In this paper, we restrict ourselves to a second order type system but it
can easily be extended to higher-order. Types are built from two basic sorts
of objects: \emph{propositions} (the types themselves) and \emph{individuals}
(untyped terms of the language). Terms appear in a restriction operator
$A \restriction t \equiv u$ and a membership predicate $t \in A$. The former
is used to define the equality types (by taking $A = \top$) and the latter is
used to encode dependent product.
$$ \Pi_{a : A} B \quad:=\quad \forall a (a \in A \Rightarrow B)$$
Overall, the higher-order version of our system is similar to a
Curry-style HOL with ML programs as individuals. It does not allow the
definition of a type which structure depends on a term (e.g. functions with
a variable number of arguments). Our system can thus be placed between HOL
(a.k.a. $F_\omega$) and the pure calculus of constructions (a.k.a. $CoC$) in
(a Curry-style and classical version of) Barendregt's $\lambda$-cube.

Throughout this paper we build a realizability model \`a la Krivine
\cite{krivine} based on a call-by-value abstract machine. As a consequence,
formulas are interpreted using three layers (values, stacks and terms)
related via orthogonality (definition \ref{orthodef}). The crucial property
(theorem \ref{biortho}) for the soundness of semantical value restriction is
that
$$ \phi^{\bot\bot} \cap \Lambda_v = \phi$$
for every set of values $\phi$ (closed under $(\equiv)$). $\Lambda_v$ denotes
the set of all values and $\phi^\bot$ (resp. $\phi^{\bot\bot}$) the set of all
stacks (resp. terms) that are compatible with every value in $\phi$
(resp. stacks in $\phi^\bot$). To obtain a model satisfying this property, we
need to extend our programming language with a term $\delta_{v,w}$ which
reduction depends on the observational equivalence of two values $v$ and $w$.

\subsection*{Related work}

To our knowledge, combining call-by-value evaluation, side-effects and
dependent products has never been achieved before. At least not for a
dependent product fully compatible with effects and call-by-value. For
example, the Aura language \cite{aura} forbids dependency on terms that
are not values in dependent applications. Similarly, the $F^\star$ language
\cite{fstar} relies on (partial) let-normal forms to enforce values
in argument position. Daniel Licata and Robert Harper have defined a notion
of positively dependent types \cite{licata} which only allow dependency over
strictly positive types.
Finally, in language like ATS \cite{ats} and DML \cite{dml} dependent types
are limited to a specific index language.

The system that seems the most similar to ours is NuPrl \cite{nuprl},
although it is inconsistent with classical reasoning. NuPrl accommodates
an observational equivalence $(\sim)$ (Howe's ``squiggle''
relation \cite{howe}) similar to our $(\equiv)$ relation. It is partially
reflected in the syntax of the system. Being based on a Kleene style
realizability model, NuPrl can also be used to reason about untyped terms.

The central part of this paper consists in a classical realizability model
construction in the style of Jean-Louis Krivine \cite{krivine}. We rely on a
call-by-value presentation which yields a model in three layers (values,
terms and stacks). Such a technique has already been used to account for
classical ML-like polymorphism in call-by-value in the work of Guillaume
Munch-Maccagnoni \cite{munch}\footnote{Our theorem \ref{biortho} seems
unrelated to lemma 9 in Munch-Maccagnoni's work \cite{munch}.}. It is here
extended to include dependent products.

The most actively developed proof assistants following the Curry-Howard
correspondence are Coq and Agda \cite{coq,agda}. The former is based on
Coquand and Huet's calculus of constructions and the latter on Martin-Löf's
dependent type theory \cite{coc,ml82}. These two constructive theories
provide dependent types, which allow the definition of very expressive
specifications. Coq and Agda do not directly give a computational
interpretation to classical logic. Classical reasoning can only be done
through the definition of axioms such as the law of the excluded middle.
Moreover, these two languages are logically consistent, and hence their
type-checkers only allow terminating programs. As termination checking is
a difficult (and undecidable) problem, many terminating programs are
rejected. Although this is not a problem for formalizing mathematics, this
makes programming tedious.

The TRELLYS project \cite{trellys} aims at providing a language in which
a consistent core can interact with type-safe dependently-typed programming
with general recursion. Although the language defined in \cite{trellys} is
call-by-value and allows effect, it suffers from value restriction like
Aura \cite{aura}. The value restriction does not appear explicitly but is
encoded into a well-formedness judgement appearing as the premise of the
typing rule for application. Apart from value restriction, the main
difference between the language of the TRELLYS project and ours resides in
the calculus itself. Their calculus is Church-style (or explicitly typed)
while ours is Curry-style (or implicitly typed). In particular, their
terms and types are defined simultaneously, while our type system is
constructed on top of an untyped calculus.

Another similar system can be found in the work of Alexandre Miquel
\cite{phdalex}, where propositions can be classical and Curry-style. However
the rest of the language remains Church style and does not embed a full
ML-like language.
The PVS system \cite{pvs96} is similar to ours as it is based on classical
higher-order logic. However this tool does not seem to be a programming
language, but rather a specification language coupled with proof checking
and model checking utilities. It is nonetheless worth mentioning that the
undecidability of PVS's type system is handled by generating proof
obligations. Our system will take a different approach and use a
non-backtracking type-checking and type-inference algorithm.

%% file: syntax.tex
The language is expressed in terms of a \emph{Krivine Abstract Machine}
\cite{kam}, which is a stack-based machine. It is formed using four
syntactic entities: values, terms, stacks and processes. The distinction
between terms and values is specific to the call-by-value presentation,
they would be collapsed in call-by-name. We require three distinct countable
sets of variables:
\begin{itemize}
\item $\mathcal{V}_\lambda = \{x, y, z...\}$ for $\lambda$-variables,
\item $\mathcal{V}_\mu = \{\alpha, \beta, \gamma...\}$ for $\mu$-variables
      (also called stack variables) and
\item $\mathcal{V}_\iota = \{a, b, c...\}$ for term variables. Term variables
      will be bound in formulas, but never in terms.
\end{itemize}
We also require a countable set $\mathcal{L} = \{l, l_1, l_2...\}$ of
labels to name record fields and a countable set $\mathcal{C} = \{C, C_1,
C_2...\}$ of constructors.
\begin{definition}
Values, terms, stacks and processes are mutually inductively defined by the
following grammars. The names of the corresponding sets are displayed on the
right.
\begin{align*}
  v,w \bnfeq & x \bnfor \lambda x\;t \bnfor C[v]
      \bnfor \{l_i = v_i\}_{i \in I}
      &(\Lambda_v)\\
  t,u \bnfeq & a
      \bnfor v \bnfor t\;u \bnfor \mu \alpha\;t \bnfor p \bnfor v.l
      \bnfor \tcase{v}{C_i[x_i] \to t_i}_{i \in I}
      \bnfor \delta_{v,w}
      &(\Lambda)\\ 
  \pi,\rho \bnfeq & \alpha \bnfor v.\pi \bnfor [t]\pi
      &(\Pi)\\
  p,q \bnfeq & t \ast \pi
      &\hspace{-0.1em}(\Lambda \times \Pi)
\end{align*}
\end{definition}
Terms and values form a variation of the $\lambda\mu$-calculus \cite{parigot}
enriched with ML-like constructs (i.e. records and variants). For technical
purposes that will become clear later on, we extend the language with a
special kind of term $\delta_{v,w}$. It will only be used to build the
model and is not intended to be accessed directly by the user. One may note
that values and processes are terms. In particular, a process of the form
$t \ast \alpha$ corresponds exactly to a named term $[\alpha]t$ in the
most usual presentation of the $\lambda\mu$-calculus. A stack can
be either a stack variable, a value pushed on top of a stack, or a stack
frame containing a term on top of a stack. These two constructors are
specific to the call-by-value presentation, only one would be required
in call-by-name.

\begin{remark}
We enforce values in constructors, record fields, projection and case
analysis. This makes the calculus simpler because only $\beta$-reduction
will manipulate the stack. We can define syntactic sugars such as the
following to hide the restriction from the programmer.
$$ t.l\;:=\;(\lambda x\,x.l)\,t \hspace{2cm} C[t]\;:=\;(\lambda x\,C[x])\,t $$
\end{remark}
\begin{definition}
Given a value, term, stack or process $\psi$ we denote $FV_\lambda(\psi)$
(resp. $FV_\mu(\psi)$, $TV(\psi)$) the set of free $\lambda$-variables (resp.
free $\mu$-variables, term variables) contained in $\psi$. We say that $\psi$
is closed if it does not contain any free variable of any kind. The set of
closed values and the set of closed terms are denoted $\Lambda_v^\ast$
and $\Lambda^\ast$ respectively.
\end{definition}
\begin{remark}
A stack, and hence a process, can never be closed as they always at least
contain a stack variable.
\end{remark}

\subsection{Call-by-value reduction relation}

Processes form the internal state of our abstract machine. They are to be
thought of as a term put in some evaluation context represented using a stack.
Intuitively, the stack $\pi$ in the process $t \ast \pi$ contains the
arguments to be fed to $t$. Since we are in call-by-value the
stack also handles the storing of functions while their arguments are being
evaluated. This is why we need stack frames (i.e. stacks of the form
$[t] \pi$). The operational semantics of our language is given by a relation
$(\succ)$ over processes.
\begin{definition}
The relation $(\succ) \subseteq (\Lambda\times\Pi)^2$ is defined as the
smallest relation satisfying the following reduction rules.
\begin{align*}
t\;u \ast \pi
  \;\;\;\;&\succ\;\;\;
  u \ast [t] \pi\\
v \ast [t] \pi
  \;\;\;\;&\succ\;\;\;
  t \ast v.\pi\\
\lambda x\;t \ast v.\pi
  \;\;\;\;&\succ\;\;\;
  t[x \!:=\! v] \ast \pi\\
\mu \alpha\;t \ast \pi
  \;\;\;\;&\succ\;\;\;
  t[\alpha \!:=\! \pi] \ast \pi\\
p \ast \pi
  \;\;\;\;&\succ\;\;\;
  p\\
  \{l_i = v_i\}_{i \in I}.l_k \ast \pi
  \;\;\;\;&\succ\;\;\;
  v_k \ast \pi & {k \in I}\\
\tcase{C_k[v]}{C_i[x_i] \to t_i}_{i \in I} \ast \pi
  \;\;\;\;&\succ\;\;\;
t_k[x_k \!:=\! v] \ast \pi & {k \in I}
\end{align*}
We will denote $(\succ^{+})$ its transitive closure, $(\succ^{*})$ its
reflexive-transitive closure and $(\succ^k)$ its $k$-fold application.
\end{definition}
The first three rules are those that handle $\beta$-reduction. When the
abstract machine encounters an application, the function is stored in a
stack-frame in order to evaluate its argument first. Once the argument has
been completely computed, a value faces the stack-frame containing the
function. At this point the function can be evaluated and the value is stored
in the stack ready to be consumed by the function as soon as it evaluates to a
$\lambda$-abstraction. A capture-avoiding substitution can then be performed
to effectively apply the argument to the function. The fourth and fifth rules
rules handle the classical part of computation. When a $\mu$-abstraction is
reached, the current stack (i.e. the current evaluation context) is captured
and substituted for the corresponding $\mu$-variable. Conversely, when a process
is reached, the current stack is thrown away and evaluation resumes with the
process. The last two rules perform projection and case analysis in the
expected way. Note that for now, states of the form $\delta_{v,w} \ast \pi$
are unaffected by the reduction relation.
\begin{remark}
For the abstract machine to be simpler, we use right-to-left call-by-value
evaluation, and not the more usual left-to-right call-by-value evaluation.
\end{remark}

\begin{lemma}\label{redcompatall}
The reduction relation $(\succ)$ is compatible with substitutions of variables
of any kind. More formally, if $p$ and $q$ are processes such that $p \succ q$
then:
\begin{itemize}
\item for all $x \in \mathcal{V}_\lambda$ and $v \in \Lambda_v$,
  $p[x := v] \succ q[x := v]$,
\item for all $\alpha \in \mathcal{V}_\mu$ and $\pi \in \Pi$,
  $p[\alpha := \pi] \succ q[\alpha := \pi]$,
\item for all $a \in \mathcal{V}_\iota$ and $t \in \Lambda$,
  $p[a := t] \succ q[a := t]$.
\end{itemize}
Consequently, if $\sigma$ is a substitution for variables of any kind and if
$p \succ q$ (resp. $p \succ^{*} q$, $p \succ^{+} q$, $p \succ^k q$) then 
$p\sigma \succ q\sigma$ (resp. $p\sigma \succ^{*} q\sigma$,
$p\sigma \succ^{+} q\sigma$, $p\sigma \succ^k q\sigma$).
\end{lemma}
\begin{proof}
  Immediate case analysis on the reduction rules.
\end{proof}
\medskip

We are now going to give the vocabulary that will be used to describe some
specific classes of processes. In particular we need to identify processes
that are to be considered as the evidence of a successful computation, and
those that are to be recognised as expressing failure.
\begin{definition}
A process $p \in \Lambda\times\Pi$ is said to be:
\begin{itemize}
\item \emph{final} if there is a value $v \in \Lambda_v$ and a stack variable
  $\alpha \in \mathcal{V}_\mu$ such that $p = v \ast \alpha$,
\item \emph{$\delta$-like} if there are values $v, w \in \Lambda_v$ and a
  stack $\pi \in \Pi$ such that $p = \delta_{v,w} \ast \pi$,
\item \emph{blocked} if there is no $q \in \Lambda\times\Pi$ such that
  $p \succ q$,
\item \emph{stuck} if it is not final nor $\delta$-like, and if for every
  substitution $\sigma$, $p\sigma$ is blocked,
\item \emph{non-terminating} if there is no blocked process
  $q \in \Lambda\times\Pi$ such that $p \succ^{*} q$.
\end{itemize}
\end{definition}

\begin{lemma}\label{redstable}
Let $p$ be a process and $\sigma$ be a substitution for variables of any kind.
If $p$ is $\delta$-like (resp. stuck, non-terminating) then $p\sigma$ is also
$\delta$-like (resp. stuck, non-terminating).
\end{lemma}
\begin{proof}
  Immediate by definition.
\end{proof}

\begin{lemma}\label{remark}
A stuck state is of one of the following forms, where $k \notin I$.
\begin{center}
\begin{tabular}{c}
$
C[v].l \ast \pi
\quad\quad\quad
(\lambda x\;t).l \ast \pi
\quad\quad\quad
C[v] \ast w.\pi
\quad\quad\quad
\{l_i = v_i\}_{i \in I} \ast v.\pi
$
\medskip
\\
$
\tcase{\lambda x\;t}{C_i[x_i] \to t_i}_{i \in I} \ast \pi
\quad\quad\quad
\tcase{\{l_i = v_i\}_{i \in I}}{C_j[x_j] \to t_j}_{j \in J} \ast \pi
$
\medskip
\\
$
\tcase{C_k[v]}{C_i[x_i] \to t_i}_{i \in I} \ast \pi
\quad\quad\quad
\{l_i = v_i\}_{i \in I}.l_k \ast \pi
$
\end{tabular}
\end{center}
\end{lemma}
\begin{proof}
  Simple case analysis.
\end{proof}

\begin{lemma}\label{possibilities}
A blocked process $p \in \Lambda\times\Pi$ is either stuck, final,
$\delta$-like, or of one of the following forms.
\begin{center}
\begin{tabular}{c}
$
x.l \ast \pi
\quad\quad\quad
x \ast v.\pi
\quad\quad\quad
\tcase{x}{C_i[x_i] \to t_i}_{i \in I} \ast \pi
\quad\quad\quad
a \ast \pi
$
\end{tabular}
\end{center}
\end{lemma}
\begin{proof}
Straight-forward case analysis using lemma \ref{remark}.
\end{proof}

\subsection{Reduction of $\delta_{v,w}$ and equivalence}

The idea now is to define a notion of observational equivalence over terms
using a relation $(\equiv)$. We then extend the reduction relation with a rule
reducing a state of the form $\delta_{v,w} \ast \pi$ to $v \ast \pi$ if
$v \not\equiv w$. If $v \equiv w$ then $\delta_{v,w}$ is stuck. With this rule
reduction and equivalence will become interdependent as equivalence will be
defined using reduction.
\begin{definition}
Given a reduction relation $R$, we say that a process $p \in \Lambda\times\Pi$
converges, and write $p \converge_R$, if there is a final state
$q \in \Lambda\times\Pi$ such that $p R^{*} q$ (where $R^{*}$ is the
reflexive-transitive closure of $R$). If $p$ does not converge we say that it
diverges and write $p \diverge_R$. We will use the notations $p \converge_i$
and $p \diverge_i$ when working with indexed notation symbols like $(\red_i)$.
\end{definition}

\begin{definition}
For every natural number $i$ we define a reduction relation $(\red_i)$ and an
equivalence relation $(\equiv_i)$ which negation will be denoted
$(\not\equiv_i)$.
$$ (\red_i) = (\succ) \cup \{(\delta_{v,w} \ast \pi, v \ast \pi) \;|\;
  \exists j < i, v \not\equiv_j w\} $$
$$ (\equiv_i) = \{(t,u) \;|\; \forall j \leq i, \forall \pi, \forall
  \sigma, {t\sigma \ast \pi \converge_j} \Leftrightarrow 
  {u\sigma \ast \pi \converge_j}\} $$
\end{definition}
It is easy to see that $(\red_0) = (\succ)$. For every natural number $i$,
the relation $(\equiv_i)$ is indeed an equivalence relation as it can be
seen as an intersection of equivalence relations. Its negation can be
expressed as follows.
$$ (\not\equiv_i) = \{(t,u), (u,t) \;|\; \exists j \leq i, \exists \pi,
  \exists \sigma, {t\sigma \ast \pi \converge_j} \land
  {u\sigma \ast \pi \diverge_j}\} $$

\begin{definition}
We define a reduction relation $(\red)$ and an equivalence relation $(\equiv)$
which negation will be denoted $(\not\equiv)$.
$$ (\red) = \bigcup_{i\in\bbN} {(\red_i)}
   \;\;\;\;\;\;
   \;\;\;\;\;\;
   (\equiv) = \bigcap_{i\in\bbN} {(\equiv_i)} $$
\end{definition}
These relations can be expressed directly (i.e. without the need of a union or
an intersection) in the following way.
\begin{align*}
(\equiv) &= \{(t,u) \;|\; \forall i, \forall \pi, \forall \sigma,
  {t\sigma \ast \pi \converge_i} \Leftrightarrow 
  {u\sigma \ast \pi \converge_i}\}\\
(\not\equiv) &= \{(t,u), (u,t) \;|\; \exists i, \exists \pi, \exists \sigma,
  {t\sigma \ast \pi \converge_i} \land {u\sigma \ast \pi \diverge_i}\}\\
(\red) &= (\succ) \cup \{(\delta_{v,w} \ast \pi, v \ast \pi) \;|\;
  v \not\equiv w\}
\end{align*}

\begin{remark}
Obviously $(\red_i) \subseteq (\red_{i+1})$ and
$(\equiv_{i+1}) \subseteq (\equiv_i)$.
As a consequence the construction of $(\red_i)_{i\in\bbN}$ and
$(\equiv_i)_{i\in\bbN}$ converges. In fact $(\red)$ and $(\equiv)$ form a
fixpoint at ordinal $\omega$. Surprisingly, this property is not explicitly
required.
\end{remark}

\begin{theorem}\label{equivpole}
Let $t$ and $u$ be terms. If $t \equiv u$ then for every stack $\pi \in \Pi$
and substitution $\sigma$ we have $t\sigma \ast \pi \converge_{\red}
\Leftrightarrow u\sigma \ast \pi \converge_{\red}$.
\end{theorem}
\begin{proof}
We suppose that $t \equiv u$ and we take $\pi_0 \in \Pi$ and a substitution
$\sigma_0$. By symmetry we can assume that
${t\sigma_0 \ast \pi_0} \converge_\red$ and show that
${u\sigma_0 \ast \pi_0} \converge_\red$. By definition there is $i_0 \in \bbN$
such that ${t\sigma_0 \ast \pi_0} \converge_{i_0}$. Since $t \equiv u$ we know
that for every $i \in \bbN$, $\pi \in \Pi$ and substitution $\sigma$ we have
${t\sigma \ast \pi} \converge_i \Leftrightarrow {u\sigma \ast \pi}
\converge_i$. This is true in particular for $i = i_0$, $\pi = \pi_0$ and
$\sigma = \sigma_0$. We hence obtain ${u\sigma_0 \ast \pi_0}
\converge_{i_0}$ which give us ${u\sigma_0 \ast \pi_0} \converge_\red$.
\end{proof}
\begin{remark}
The converse implication is not true in general: taking
$t = \delta_{\lambda x\,x,\{\}}$ and $u = \lambda x\,x$ gives a
counter-example. More generally
${p\converge_\red} \Leftrightarrow {q\converge_\red}$ does not necessarily
imply  ${p\converge_i} \Leftrightarrow {q\converge_i}$ for all
$i\in\mathbb{N}$. 
\end{remark}
\begin{corollary}\label{eqconvconv}
Let $t$ and $u$ be terms and $\pi$ be a stack. If $t \equiv u$ and
${t \ast \pi} \converge_\red$ then ${u \ast \pi} \converge_\red$.
\end{corollary}
\begin{proof}
Direct consequence of theorem \ref{equivpole} using $\pi$ and an empty
substitution.
\end{proof}

\subsection{Extensionality of the language}

In order to be able to work with the equivalence relation $(\equiv)$, we need
to check that it is extensional. In other words, we need to be able to replace
equals by equals at any place in terms without changing their observed
behaviour. This property is summarized in the following two theorems.

\begin{theorem}\label{extval}
Let $v$ and $w$ be values, $E$ be a term and $x$ be a $\lambda$-variable. If
$v \equiv w$ then $E[x := v] \equiv E[x := w]$.
\end{theorem}
\begin{proof}
We are going to prove the contrapositive so we suppose
$E[x := v] \not\equiv E[x := w]$ and show $v \not\equiv w$. By definition
there is $i \in \bbN$, $\pi \in \Pi$ and a substitution $\sigma$ such
that $(E[x := v])\sigma \ast \pi \converge_i$ and
$(E[x := w])\sigma \ast \pi \diverge_i$ (up to symmetry). Since we can
rename $x$ in such a way that it does not appear in $dom(\sigma)$, we can
suppose $E\sigma[x := v\sigma] \ast \pi \converge_i$ and
$E\sigma[x := w\sigma] \ast \pi \diverge_i$.
In order to show $v \not\equiv w$ we need to find $i_0 \in \bbN$,
$\pi_0 \in \Pi$ and a substitution $\sigma_0$ such that
$v\sigma_0 \ast \pi_0 \converge_{i_0}$ and
$w\sigma_0 \ast \pi_0 \diverge_{i_0}$ (up to symmetry). We take $i_0 = i$,
$\pi_0 = [\lambda x\;E\sigma]\pi$ and $\sigma_0 = \sigma$. These
values are suitable since by definition
${v\sigma_0 \ast \pi_0 } \red_{i_0} {E\sigma[x := v\sigma] \ast \pi}
\converge_{i_0}$ and
${w\sigma_0 \ast \pi_0} \red_{i_0} {E\sigma[x := w\sigma] \ast \pi}
\diverge_{i_0}$.
\end{proof}

\begin{lemma}\label{aposs}
Let $s$ be a process, $t$ be a term, $a$ be a term variable and $k$ be a
natural number. If $s[a := t] \converge_k$ then there is a blocked state
$p$ such that $s \succ^{*} p$ and either
\begin{itemize}
\item $p = v \ast \alpha$ for some value $v$ and a stack variable $\alpha$,
\item $p = a \ast \pi$ for some stack $\pi$,
\item $k > 0$ and $p = \delta(v,w) \ast \pi$ for some values $v$ and $w$ and
      stack $\pi$, and in this case $v[a := t] \not\equiv_j w[a := t]$ for
      some $j < k$.
\end{itemize}
\end{lemma}
\begin{proof}
Let $\sigma$ be the substitution $[a := t]$. If $s$ is non-terminating, lemma
\ref{redstable} tells us that $s\sigma$ is also non-terminating, which
contradicts $s\sigma \converge_k$. Consequently, there is a blocked process
$p$ such that $s \succ^{*} p$ since $(\succ) \subseteq (\red_k)$. Using lemma
\ref{redcompatall} we get
$s\sigma \succ^{*} p\sigma$ from which we obtain $p\sigma \converge_{k}$.
The process $p$ cannot be stuck, otherwise $p\sigma$ would also be stuck
by lemma \ref{redstable}, which would contradict $p\sigma \converge_{k}$.
Let us now suppose that $p = \delta_{v,w} \ast \pi$ for some values $v$ and
$w$ and some stack $\pi$. Since
$\delta_{v\sigma,w\sigma} \ast \pi \converge_k$ there must be $i < k$ such
that $v\sigma \not\equiv_j w\sigma$, otherwise this would contradict
$\delta_{v\sigma,w\sigma} \ast \pi \converge_k$. In this case we necessarily
have $k > 0$, otherwise there would be no possible candidate for $i$.
According to lemma \ref{possibilities} we need to rule out four more forms of
therms: $x.l \ast \pi$, $x \ast v.\pi$, $case_x\;B \ast \pi$ and $b \ast \pi$
in the case where $b \not= a$. If $p$ was of one of these forms the
substitution $\sigma$ would not be able to unblock the reduction of $p$, which
would contradict again $p\sigma \converge_{k}$.
\end{proof}

\begin{lemma}\label{aextlem}
Let $t_1$, $t_2$ and $E$ be terms and $a$ be a term variable. For every
$k \in \bbN$, if $t_1 \equiv_k t_2$ then $E[a \!:=\! t_1] \equiv_k E[a
\!:=\! t_2]$.
\end{lemma}
\begin{proof}
Let us take $k \in \bbN$, suppose that $t_1 \equiv_k t_2$ and show that
$E[a \!:=\! t_1] \equiv_k E[a \!:=\! t_1]$. By symmetry we can assume that
we have $i \leq k$, $\pi \in \Pi$ and a substitution $\sigma$ such that
$(E[a \!:=\! t_1])\sigma \ast \pi \converge_i$ and show that
$(E[a \!:=\! t_2])\sigma \ast \pi \converge_i$. As we are free to rename
$a$, we can suppose that it does not appear in $dom(\sigma)$, $TV(\pi)$,
$TV(t_1)$ or $TV(t_2)$. In order to lighten the notations we define
$E' = E\sigma$, $\sigma_1 = [a \!:=\! t_1\sigma]$ and
$\sigma_2 = [a \!:=\! t_2\sigma]$. We are hence assuming
$E'\sigma_1 \ast \pi \converge_i$ and trying to show
$E'\sigma_2 \ast \pi \converge_i$.

We will now build a sequence $(E_i,\pi_i,l_i)_{i \in I}$ in such a way that
$E'\sigma_1 \ast \pi \reds_k E_i\sigma_1 \ast \pi_i\sigma_1$ in $l_i$ steps
for every $i \in I$. Furthermore, we require that $(l_i)_{i \in I}$ is
increasing and that it has a strictly increasing subsequence. Under this
condition our sequence will necessarily be finite. If it was infinite the
number of reduction steps that could be taken from the state
$E'\sigma_1 \ast \pi$ would not be bounded, which would contradict
$E'\sigma_1 \ast \pi \converge_i$. We now denote our finite sequence
$(E_i,\pi_i,l_i)_{i \leq n}$ with $n \in \bbN$. In order to show that
$(l_i)_{i \leq n}$ has a strictly increasing subsequence, we will ensure that
it does not have three equal consecutive values. More formally, we will
require that if $0 < i < n$ and $l_{i-1} = l_i$ then $l_{i+1} > l_i$.

To define $(E_0,\pi_0,l_0)$ we consider the reduction of $E' \ast \pi$. Since
we know that $(E' \ast \pi)\sigma_1 = E'\sigma_1 \ast \pi \converge_i$ we
use lemma \ref{aposs} to obtain a blocked state $p$ such that
${E' \ast \pi} \succ^j p$.
We can now take $E_0 \ast \pi_0 = p$ and $l_0 = j$.
By lemma \ref{redcompatall} we have
$(E' \ast \pi)\sigma_1 \succ^j {E_0\sigma_1 \ast \pi_0\sigma_1}$ from which we
can deduce that
$(E' \ast \pi)\sigma_1 \reds_k {E_0\sigma_1 \ast \pi_0\sigma_1}$ in $l_0 = j$
steps.

To define $(E_{i+1},\pi_{i+1},l_{i+1})$ we consider the reduction of the
process $E_i\sigma_1 \ast \pi_i$. By construction we know that
${E'\sigma_1 \ast \pi}
\reds_k {E_i\sigma_1 \ast \pi_i\sigma_1 = (E_i\sigma_1 \ast \pi_i)\sigma_1}$
in $l_i$ steps. Using lemma \ref{aposs} we know that $E_i \ast \pi_i$ might be
of three shapes.
\begin{itemize}
\item If ${E_i \ast \pi_i} = {v \ast \alpha}$ for some value $v$ and stack
      variable $\alpha$ then the end of the sequence was reached with $n = i$.
\item If $E_i = a$ then we consider the reduction of $E_i\sigma_1 \ast \pi_i$.
      Since $(E_i\sigma_1 \ast \pi_i)\sigma_1 \converge_k$ we know from
      lemma \ref{aposs} that there is a blocked process $p$ such that
      ${E_i\sigma_1 \ast \pi_i} \succ^j p$. Using lemma \ref{redcompatall} we
      obtain that ${E_i\sigma_1 \ast \pi_i\sigma_1} \succ^j p\sigma_1$ from which we
      can deduce that ${E_i\sigma_1 \ast \pi_i\sigma_1} \red_k p\sigma_1$ in
      $j$ steps. We then take $E_{i+1} \ast \pi_{i+1} = p$ and
      $l_{i+1} = l_i + j$.

      Is it possible to have $j=0$? This can only happen when
      $E_i\sigma_1 \ast \pi_i$ is of one of the three forms of lemma
      \ref{aposs}. It cannot be of the form $a \ast \pi$ as we assumed that
      $a$ does not appear in $t_1$ or $\sigma$. If it is of the form $v \ast
      \alpha$, then we reached the end of the sequence with $i + 1 = n$ so
      there is no trouble. The process $E_i\sigma_1 \ast \pi_i$ may be of the
      form $\delta(v,w) \ast \pi$, but we will have $l_{i+2} > l_{i+1}$.
\item If $E_i = \delta(v,w)$ for some values $v$ and $w$ we have
      $m < k$ such that $v\sigma_1 \not\equiv_m w\sigma_1$. Hence
      ${E_i\sigma_1 \ast \pi_i = \delta(v\sigma_1,w\sigma_1) \ast \pi_i}
      \red_k {v\sigma_1 \ast \pi_i}$ by definition. Moreover
      ${E_i\sigma_1 \ast \pi_i\sigma_1} \red_k {v\sigma_1 \ast \pi_i\sigma_1}$
      by lemma \ref{redcompatall}. Since
      ${E'\sigma_1 \ast \pi} \reds_k {E_i\sigma_1 \ast \pi_i\sigma_1}$ in
      $l_i$ steps we obtain that
      ${E'\sigma_1 \ast \pi} \reds_k {v\sigma_1 \ast \pi_i\sigma_1}$ in
      $l_i + 1$ steps. This also gives us ${(v\sigma_1 \ast \pi_i)\sigma_1 = 
      v\sigma_1 \ast \pi_i\sigma_1} \converge_k$.
      
      We now consider the reduction of the process $v\sigma_1 \ast \pi_i$. By
      lemma \ref{aposs} there is a blocked process $p$ such that
      ${v\sigma_1 \ast \pi_i} \succ^j p$. Using lemma \ref{redcompatall} we
      obtain ${v\sigma_1 \ast \pi_i\sigma_1} \succ^j p\sigma_1$ from which we
      deduce that ${v\sigma_1 \ast \pi_i\sigma_1} \reds_k p\sigma_1$ in $j$
      steps. We then take $E_{i+1} \ast \pi_{i+1} = p$ and
      $l_{i+1} = l_i + j + 1$. Note that in this case we have $l_{i+1} > l_i$.
\end{itemize}
Intuitively $(E_i,\pi_i,l_i)_{i \leq n}$ mimics the reduction of
$E'\sigma_1 \ast \pi$ while making explicit every substitution of $a$ and
every reduction of a $\delta$-like state.

To end the proof we show that for every $i \leq n$ we have
${E_i\sigma_2 \ast \pi_i\sigma_2} \converge_k$. For $i = 0$ this will give
us ${E'\sigma_2 \ast \pi} \converge_k$ which is the expected result. Since
$E_n \ast \pi_n = v \ast \alpha$ we have $E_n\sigma_2 \ast \pi_n\sigma_2 =
v\sigma_2 \ast \alpha$ from which we trivially obtain
${E_n\sigma_2 \ast \pi_n\sigma_2} \converge_k$.
We now suppose that ${E_{i+1}\sigma_2 \ast \pi_i\sigma_2} \converge_k$ for
$0 \leq i < n$ and show that ${E_i\sigma_2 \ast \pi_i\sigma_2} \converge_k$.
By construction $E_i \ast \pi_i$ can be of two shapes\footnote{Only
$E_n \ast \pi_n$ can be of the form $v \ast \alpha$.}:
\begin{itemize}
\item If $E_i=a$ then ${t_1\sigma \ast \pi_i} \reds_k {E_{i+1} \ast \pi_{i+1}}$.
  Using lemma \ref{redcompatall} we obtain
  $t_1\sigma \ast \pi_i\sigma_2 \red_k E_{i+1}\sigma_2 \ast
  \pi_i\sigma_2$ from which we deduce $t_1\sigma \ast \pi_i\sigma_2
  \converge_k$ by induction hypothesis. Since $t_1 \equiv_k t_2$ we obtain
  ${t_2\sigma \ast \pi_i\sigma_2 = (E_i \ast \pi_i)\sigma_2} \converge_k$.
\item If $E_i = \delta(v,w)$ then ${v \ast \pi_i} \red_k {E_{i+1} \ast
  \pi_{i+1}}$ and hence $v\sigma_2 \ast \pi_i\sigma_2 \red_k
  E_{i+1}\sigma_2 \ast \pi_{i+1}\sigma_2$ by lemma \ref{redcompatall}. Using the
  induction hypothesis we obtain ${v\sigma_2 \ast \pi_i\sigma_2} \converge_k$.
  It remains to show that ${\delta(v\sigma_2,w\sigma_2) \ast \pi_i\sigma_2}
  \reds_k {v\sigma_2 \ast \pi_i\sigma_2}$.
  We need to find $j < k$ such that $v\sigma_2 \not\equiv_j w\sigma_2$.
  By construction there is $m < k$ such that
  $v\sigma_1 \not\equiv_m w\sigma_1$. We are going to show that
  $v\sigma_2 \not\equiv_m w\sigma_2$. By using the global induction
  hypothesis twice we obtain $v\sigma_1 \equiv_m v\sigma_2$ and
  $w\sigma_1 \equiv_m v\sigma_2$. Now if $v\sigma_2 \equiv_m w\sigma_2$ then
  $v\sigma_1 \equiv_m v\sigma_2 \equiv_m w\sigma_2 \equiv_m w\sigma_1$
  contradicts $v\sigma_1 \not\equiv w\sigma_1$. Hence we must have
  $v\sigma_2 \not\equiv_m w\sigma_2$.
\end{itemize}
\end{proof}

\begin{theorem}\label{extterm}
Let $t_1$, $t_2$ and $E$ be three terms and $a$ be a term variable. If
$t_1 \equiv t_2$ then $E[a \!:=\! t_1] \equiv E[a \!:=\! t_2]$.
\end{theorem}
\begin{proof}
We suppose that $t_1 \equiv t_2$ which means that $t_1 \equiv_i t_2$ for
every $i \in \bbN$. We need to show that
$E[a \!:=\! t_1] \equiv E[a \!:=\! t_2]$ so we take $i_0 \in \bbN$ and show
$E[a \!:=\! t_1] \equiv_{i_0} E[a \!:=\! t_2]$. By hypothesis we have
$t_1 \equiv_{i_0} t_2$ and hence we can conclude using lemma \ref{aextlem}.
\end{proof}

%% file: semantics.tex
The syntax presented in the previous section is part of a realizability
machinery that will be built upon here. We aim at obtaining a semantical
interpretation of the second-order type system that will be defined shortly.
Our abstract machine slightly differs from the mainstream presentation of
\emph{Krivine's classical realizability} which is usually call-by-name.
Although call-by-value presentations have rarely been published, such
developments are well-known among classical realizability experts. The
addition of the $\delta$ instruction and the related modifications are
however due to the author.

\subsection{Pole and orthogonality}

As always in classical realizability, the model is parametrized by a pole,
which serves as an exchange point between the world of programs and the
world of execution contexts (i.e. stacks).
\begin{definition}
A \emph{pole} is a set of processes $\dbot \subseteq \Lambda\times\Pi$ which
is \emph{saturated} (i.e. closed under backward reduction). More formally,
if we have $q \in \dbot$ and $p \red q$ then $p \in \dbot$.
\end{definition}
Here, for the sake of simplicity and brevity, we are only going to use the
pole
$$ \dbot = \{p \in \Lambda\times\Pi \;|\; p \converge_\red \} $$
which is clearly saturated. Note that this particular pole is also closed
under the reduction relation $(\red)$, even though this is not a general
property. In particular $\dbot$ contains all final processes.

The notion of \emph{orthogonality} is central in Krivine's classical
realizability. In this framework a type is interpreted (or realized) by
programs computing corresponding values. This interpretation is spread in a
three-layered
construction, even though it is fully determined by the first layer (and the
choice of the pole). The first layer consists of a set of values that we will
call the \emph{raw semantics}. It gathers all the syntactic values that should
be considered as having the corresponding type. As an example, if we were to
consider the type of natural numbers, its raw semantics would be the set
$\{\bar{n} \;|\; n \in \bbN\}$ where $\bar{n}$ is some encoding of $n$. The
second layer, called \emph{falsity value} is a set containing every stack that
is a candidate for building a valid process using any value from the raw
semantics. The notion of validity depends on the choice of the pole. Here for
instance, a valid process is a normalizing one (i.e. one that reduces to a
final state). The third layer, called \emph{truth value} is a set of terms
that is built by iterating the process once more. The formalism for the two
levels of orthogonality is given in the following definition.
\begin{definition}\label{orthodef}
For every set $\phi \subseteq \Lambda_v$ we define a set
$\phi^\bot \subseteq \Pi$ and a set $\phi^{\bot\bot} \subseteq \Lambda$ as
follows.
\begin{align*}
  \phi^\bot &=
    \{\pi \in \Pi \;|\; \forall v \in \phi, v \ast \pi \in \dbot\}\\
  \phi^{\bot\bot} &=
    \{t \in \Lambda \;|\; \forall \pi \in \phi^\bot, t \ast \pi \in \dbot\}
\end{align*}
\end{definition}

We now give two general properties of orthogonality that are true in every
classical realizability model. They will be useful when proving the soundness
of our type system.
\begin{lemma}\label{orthosimple}
  If $\phi \subseteq \Lambda_v$ is a set of values, then $\phi \subseteq
  \phi^{\bot\bot}$.
\end{lemma}
\begin{proof}
  Immediate following the definition of $\phi^{\bot\bot}$.
\end{proof}
\begin{lemma}\label{orthoblabla}
  Let $\phi \subseteq \Lambda_v$ and $\psi \subseteq \Lambda_v$ be two
  sets of values. If $\phi \subseteq \psi$ then
  $\phi^{\bot\bot} \subseteq \psi^{\bot\bot}$.
\end{lemma}
\begin{proof}
  Immediate by definition of orthogonality.
\end{proof}
\medskip

The construction involving the terms of the form $\delta_{v,x}$ and
$(\equiv)$ in the previous section is now going to gain meaning. The
following theorem, which is our central result, does not hold in every
classical realizability model. Obtaining a proof required us to internalize
observational equivalence, which introduces a non-computable reduction
rule.
\begin{theorem}\label{biortho}
If $\Phi \subseteq \Lambda_v$ is a set of values closed under $(\equiv)$, then
$\Phi^{\bot\bot} \cap \Lambda_v = \Phi$.
\end{theorem}
\begin{proof}
The direction $\Phi \subseteq \Phi^{\bot\bot} \cap \Lambda_v$ is
straight-forward using lemma \ref{orthosimple}. We are going to
show that $\Phi^{\bot\bot} \cap \Lambda_v \subseteq \Phi$, which amounts to
showing that for every value $v \in \Phi^{\bot\bot}$ we have $v \in \Phi$.
We are going to show the contrapositive, so let us assume $v \not\in \Phi$ and
show $v \not\in \Phi^{\bot\bot}$. We need to find a stack $\pi_0$ such that
$v \ast \pi_0 \not\in \dbot$ and for every value $w \in \Phi$,
$w \ast \pi_0 \in \dbot$. We take $\pi_0 = [\lambda x\;\delta_{x,v}]\;\alpha$
and show that is is suitable. By definition of the reduction relation
$v \ast \pi_0$ reduces to $\delta_{v,v} \ast \alpha$ which is not in $\dbot$
(it is stuck as $v \equiv v$ by reflexivity). Let us now take $w \in \Phi$.
Again by definition, $w \ast \pi_0$ reduces to $\delta_{w,v} \ast \alpha$,
but this time we have $w \not\equiv v$ since $\Phi$ was supposed to be closed
under $(\equiv)$ and $v \not\in \Phi$. Hence $w \ast \pi_0$ reduces to
${w \ast \alpha} \in \dbot$.
\end{proof}
\medskip

It is important to check that the pole we chose does not yield a degenerate
model. In particular we check that no term is able to face every stacks. If
it were the case, such a term could be use as a proof of $\bot$.
\begin{theorem}\label{poleconsist}
The pole $\dbot$ is consistent, which means that for every closed term $t$
there is a stack $\pi$ such that $t \ast \pi \not\in \dbot$.
\end{theorem}
\begin{proof}
Let $t$ be a closed term and $\alpha$ be a stack constant. If we do not have
$t \ast \alpha \converge_\red$ then we can directly take $\pi = \alpha$.
Otherwise we know that $t \ast \alpha \reds v \ast \alpha$ for some value
$v$. Since $t$ is closed $\alpha$ is the only available stack variable. We now
show that $\pi = [\lambda x\;\{\}]\{\}.\beta$ is suitable. We denote $\sigma$
the substitution $[\alpha := \pi]$. Using a trivial extension of lemma
\ref{redcompatall} to the $(\red)$ relation we obtain $t \ast \pi = (t \ast
\alpha)\sigma \reds (v \ast \alpha)\sigma = v\sigma \ast \pi$. We hence have
$t \ast \pi \reds v\sigma \ast [\lambda x\;\{\}]\{\}.\beta \red^2 \{\} \ast
\{\}.\beta \not\in \dbot$.
\end{proof}
\medskip

\subsection{Formulas and their semantics}

In this paper we limit ourselves to second-order logic, even though the system
can easily be extended to higher-order. For every natural number $n$ we
require a countable set ${\mathcal{V}}_n = \{{X}_n, {Y}_n, {Z}_n ...\}$ of
$n$-ary predicate variables.
\begin{definition}
The syntax of formulas is given by the following grammar.
\begin{align*} 
A,B
  \bnfeq &{X}_n(t_1, ..., t_n)
  \bnfor A \Rightarrow B
  \bnfor \forall a\; A
  \bnfor \exists a\; A
  \bnfor \forall X_n\; A
  \bnfor \exists X_n\; A
\\
  \bnfor &\{l_i : A_i\}_{i \in I}
  \bnfor [C_i : A_i]_{i \in I}
  \bnfor t \in A
  \bnfor A \restriction t \equiv u
\end{align*}
\end{definition}
Terms appear in several places in formulas, in particular, they form the
individuals of the logic. They can be quantified over and are used as
arguments for predicate variables. Besides the ML-like formers for sums and
products (i.e. records and variants) we add a membership predicate and a
restriction
operation. The membership predicate $t \in A$ is used to express the fact that
the term $t$ has type $A$. It provides a way to encode the dependent product
type using universal quantification and the arrow type. In this sense, it is
inspired and related to Krivine's relativization of quantifiers.
$$ \Pi_{a:A}\;B \quad:=\quad \forall a (a \in A \Rightarrow B) $$
The restriction operator can be thought of as a kind of conjunction with no
algorithmic content. The formula $A \restriction t \equiv u$ is to be
interpreted in the same way as $A$ if the equivalence $t \equiv u$ holds, and
as $\bot$ otherwise\footnote{We use the standard second-order encoding:
$\bot = \forall X_0\; X_0$ and $\top = \exists X_0\; X_0$.}. In particular,
we will define the following types:
$$
A \restriction t \not\equiv u := A \restriction t \equiv u \Rightarrow \bot
\quad\quad
t \equiv u := \top \restriction t \equiv u
\quad\quad
t \not\equiv u := \top \restriction t \not\equiv u
$$

To handle free variables in formulas we will need to generalize the
notion of substitution to allow the substitution of predicate variables.
\begin{definition}
A substitution is a finite map $\sigma$ ranging over $\lambda$-variables,
$\mu$-variables, term and predicate variables such that:
\begin{itemize}
  \item if $x \in dom(\sigma)$ then $\sigma(x) \in \Lambda_v$,
  \item if $\alpha \in dom(\sigma)$ then $\sigma(\alpha) \in \Pi$,
  \item if $a \in dom(\sigma)$ then $\sigma(a) \in \Lambda$,
  \item if $X_n \in dom(\sigma)$ then
    $\sigma(X_n) \in {\Lambda^n \to \mathcal{P}({{\Lambda}_v}/\!\!\equiv)}$.
\end{itemize}
\end{definition}
\begin{remark}
A predicate variable of arity $n$ will be substituted by a $n$-ary predicate.
Semantically, such predicate will correspond to some total (set-theoretic)
function building a subset of $\Lambda_v/\!\!\equiv$ from $n$ terms. In the
syntax, the binding of the arguments of a predicate variables will happen
implicitly during its substitution.
\end{remark}
\begin{definition}
Given a formula $A$ we denote $FV(A)$ the set of its free variables. Given a
substitution $\sigma$ such that $FV(A) \subseteq dom(\sigma)$ we write
$A[\sigma]$ the closed formula built by applying $\sigma$ to $A$.
\end{definition}

In the semantics we will interpret closed formulas by sets of values closed
under the equivalence relation $(\equiv)$.
\begin{definition}
Given a formula $A$ and a substitution $\sigma$ such that $A[\sigma]$ is
closed, we define the \emph{raw semantics}
$\llbracket A \rrbracket_\sigma \subseteq \Lambda_v/\!\!\equiv$ of $A$ under
the substitution $\sigma$ as follows.
\begin{align*}
\llbracket X_n(t_1, ..., t_n) \rrbracket_\sigma =&\;
  \sigma(X_n)(t_1\sigma, ..., t_n\sigma)\\
\llbracket A \Rightarrow B \rrbracket_\sigma =&\;
  \{\lambda x\; t \;\;|\;\; \forall v \in \llbracket A \rrbracket_\sigma,
  t[x := v] \in \llbracket B \rrbracket_\sigma^{\bot\bot} \} \\
\llbracket \forall a\; A \rrbracket_\sigma =&\;
  \cap_{t \in \Lambda^\ast}{\llbracket A \rrbracket_{\sigma[a := t]}}\\
\llbracket \exists a\; A \rrbracket_\sigma =&\;
  \cup_{t \in \Lambda^\ast}{\llbracket A \rrbracket_{\sigma[a := t]}}\\
\llbracket \forall X_n\; A \rrbracket_\sigma =&\;
  \cap_{P \in \Lambda^n \to \mathcal{P}(\Lambda_v / \equiv)}
  {\llbracket A \rrbracket_{\sigma[X_n := P]}}\\
\llbracket \exists X_n\; A \rrbracket_\sigma =&\;
  \cup_{P \in \Lambda^n \to \mathcal{P}(\Lambda_v / \equiv)}
  {\llbracket A \rrbracket_{\sigma[X_n := P]}}\\
\llbracket \{l_i : A_i\}_{i \in I} \rrbracket_\sigma =&\;
  \{\{l_i = v_i\}_{i \in I} \;\;|\;\; {\forall i \in I}\;\; v_i \in
  \llbracket A_i \rrbracket_\sigma\}\\
\llbracket [C_i : A_i]_{i \in I} \rrbracket_\sigma =&\;
  \cup_{i \in I}\{C_i[v] \;\;|\;\; v \in \llbracket A_i \rrbracket_\sigma\}\\
\llbracket t \in A \rrbracket_\sigma =&\;
  \{v \in \llbracket A \rrbracket_\sigma \;\;|\;\; t\sigma \equiv v\}\\
\llbracket A \restriction t \equiv u \rrbracket_\sigma =&\;
  \left\{ 
    \begin{array}{l l}
      \llbracket A \rrbracket_\sigma & \text{if $t\sigma \equiv u\sigma$}\\
      \emptyset & \text{otherwise}
     \end{array} \right.
\end{align*}
\end{definition}

In the model, programs will realize closed formulas in two different ways
according to their syntactic class. The interpretation of values will be
given in terms of raw semantics, and the interpretation of terms in general
will be given in terms of truth values.
\begin{definition}
Let $A$ be a formula and $\sigma$ a substitution such that $A[\sigma]$ is closed.
We say that:
\begin{itemize}
\item $v \in \Lambda_v$ realizes $A[\sigma]$ if
  $v \in \llbracket A \rrbracket_\sigma$,
\item $t \in \Lambda$ realizes $A[\sigma]$ if
  $t \in \llbracket A \rrbracket_\sigma^{\bot\bot}$.
\end{itemize}
\end{definition}

\subsection{Contexts and typing rules}

Before giving the typing rules of our system we need to define contexts and
judgements. As explained in the introduction, several typing rules require a
value restriction in our context. This is reflected in typing rule by the
presence of two forms of judgements.

\begin{definition}
A context is an ordered list of hypotheses. In particular, it contains type
declarations for $\lambda$-variables and $\mu$-variables, and declaration of
term variables and predicate variables. In our case, a context also contains
term equivalences and inequivalences. A context is built using the following
grammar.
\begin{align*}
\Gamma, \Delta \bnfeq &\bullet
  \bnfor \Gamma, x : A
  \bnfor \Gamma, \alpha : \lnot A
  \bnfor \Gamma, a : Term
\\
  \bnfor &\Gamma, X_n : Pred_n
  \bnfor \Gamma, t \equiv u
  \bnfor \Gamma, t \not\equiv u
\end{align*}
A context $\Gamma$ is said to be valid if it is possible to derive
$\Gamma\;\text{Valid}$ using the rules of figure \ref{valid_context}. In
the following, every context will be considered valid implicitly.
\end{definition}
\begin{figure}
\center
\begin{prooftree}
\AxiomC{$\Gamma \;\; \text{Valid}$}
\AxiomC{$x \not\in dom(\Gamma)$}
\AxiomC{$FV(A) \subseteq dom(\Gamma) \cup \{x\}$}
\TrinaryInfC{$\Gamma, x : A \;\; \text{Valid}$}
\end{prooftree}

\begin{prooftree}
\AxiomC{$\Gamma \;\; \text{Valid}$}
\AxiomC{$\alpha \not\in dom(\Gamma)$}
\AxiomC{$FV(A) \subseteq dom(\Gamma)$}
\TrinaryInfC{$\Gamma, \alpha : \lnot A \;\; \text{Valid}$}
\end{prooftree}

\begin{prooftree}
\AxiomC{$\Gamma \;\; \text{Valid}$}
\AxiomC{$a \not\in dom(\Gamma)$}
\BinaryInfC{$\Gamma, a : Term \;\; \text{Valid}$}
\DisplayProof\quad\quad\quad
\AxiomC{$\Gamma \;\; \text{Valid}$}
\AxiomC{$X_n \not\in dom(\Gamma)$}
\BinaryInfC{$\Gamma, X_n : Pred_n \;\; \text{Valid}$}
\end{prooftree}

\begin{prooftree}
\AxiomC{$\Gamma \;\; \text{Valid}$}
\AxiomC{$FV(t) \cup FV(u) \subseteq dom(\Gamma)$}
\BinaryInfC{$\Gamma, t \equiv u \;\; \text{Valid}$}
\end{prooftree}

\begin{prooftree}
\AxiomC{$\Gamma \;\; \text{Valid}$}
\AxiomC{$FV(t) \cup FV(u) \subseteq dom(\Gamma)$}
\BinaryInfC{$\Gamma, t \not\equiv u \;\; \text{Valid}$}
\DisplayProof\quad\quad\quad
\AxiomC{\vbox{\vspace{0.96em}}} 
\UnaryInfC{$\bullet \;\; \text{Valid}$}
\end{prooftree}
\caption{Rules allowing the construction of a valid context.}
\label{valid_context}
\end{figure}

\begin{definition}
There are two forms of typing judgements:
\begin{itemize}
\item $\Gamma \vvdash v : A$ meaning that the value $v$ has type $A$ in
  context $\Gamma$,
\item $\Gamma \vdash t : A$ meaning that the term $t$ has type $A$ in
  context $\Gamma$.
\end{itemize}
\end{definition}
\input{rules.tex}

The typing rules of the system are given in figure \ref{pml2rules}. Although
most of them are fairly usual, our type system differs in several ways. For
instance the last four rules are related to the extensionality of the
calculus. One can note the value restriction in several places: both
universal quantification introduction rules and the introduction of the
membership predicate. In fact, some value restriction is also hidden in the
rules for the elimination of the existential quantifiers and the elimination
rule for the restriction connective. These rules are presented in their
left-hand side variation, and only values can appear on the left of the
sequent. It is not surprising that elimination of an existential quantifier
requires value restriction as it is the dual of the introduction rule of a
universal quantifier.

An important and interesting difference with existing type systems
is the presence of $\uparrow$ and $\downarrow$. These two rules
allow one to go from one kind of sequent to the other when working on values.
Going from $\Gamma \vvdash v : A$ to $\Gamma \vdash v : A$ is
straight-forward. Going the other direction is the main motivation for our
model. This allows us to lift the value restriction expressed in the syntax
to a restriction expressed in terms of equivalence. For example, the two rules
\begin{prooftree}
\AxiomC{$\Gamma, t \equiv v \vdash t : A$}
\AxiomC{$a \not\in FV(\Gamma)$}
\RightLabel{$\forall_{i,\equiv}$}
\BinaryInfC{$\Gamma, t \equiv v \vdash t : \forall a\;A$}
\end{prooftree}
\begin{prooftree}
\AxiomC{$\Gamma, u \equiv v \vdash t : \Pi_{a:A} B$}
\AxiomC{$\Gamma, u \equiv v \vdash u : A$}
\RightLabel{$\Pi_{e,\equiv}$}
\BinaryInfC{$\Gamma, u \equiv v \vdash t\,u : B[a := u]$}
\end{prooftree}
can be derived in the system (see figure \ref{deriv}).
The value restriction can be removed similarly on every other rule. Thus,
judgements on values can be completely ignored by the user of the system.
Transition to value judgements will only happen internally.

\begin{figure}
\begin{prooftree}
\AxiomC{$\Gamma, t \equiv v \vdash t : A$}
\RightLabel{$\equiv_{t,l}$}
\UnaryInfC{$\Gamma, t \equiv v \vdash v : A$}
\RightLabel{$\downarrow$}
\UnaryInfC{$\Gamma, t \equiv v \vvdash v : A$}
\AxiomC{$a \not\in FV(\Gamma)$}
\RightLabel{$\forall_i$}
\BinaryInfC{$\Gamma, t \equiv v \vvdash v : \forall a\;A$}
\RightLabel{$\uparrow$}
\UnaryInfC{$\Gamma, t \equiv v \vdash v : \forall a\;A$}
\RightLabel{$\equiv_{t,l}$}
\UnaryInfC{$\Gamma, t \equiv v \vdash t : \forall a\;A$}
\end{prooftree}
\begin{prooftree}
\alwaysDoubleLine
\AxiomC{$\Gamma, u \equiv v \vdash t : \Pi_{a:A} B$}
\UnaryInfC{$\Gamma, u \equiv v \vdash t : \forall a (a \in A \Rightarrow B)$}
\alwaysSingleLine
\RightLabel{$\forall_e$}
\UnaryInfC{$\Gamma, u \equiv v \vdash t : u \in A \Rightarrow B[a := u]$}
\AxiomC{$\Gamma, u \equiv v \vdash u : A$}
\RightLabel{$\equiv_{t,l}$}
\UnaryInfC{$\Gamma, u \equiv v \vdash v : A$}
\RightLabel{$\downarrow$}
\UnaryInfC{$\Gamma, u \equiv v \vvdash v : A$}
\RightLabel{$\in_i$}
\UnaryInfC{$\Gamma, u \equiv v \vvdash v : v \in A$}
\RightLabel{$\uparrow$}
\UnaryInfC{$\Gamma, u \equiv v \vdash v : v \in A$}
\RightLabel{$\equiv_{t,l}$}
\UnaryInfC{$\Gamma, u \equiv v \vdash u : v \in A$}
\RightLabel{$\equiv_{t,r}$}
\UnaryInfC{$\Gamma, u \equiv v \vdash u : u \in A$}
\RightLabel{$\Rightarrow_e$}
\BinaryInfC{$\Gamma, u \equiv v \vdash t\,u : B[a := u]$}
\end{prooftree}
\caption{Derivation of the rules $\forall_{i,\equiv}$ and $\Pi_{e,\equiv}$.}
\label{deriv}
\end{figure}

\subsection{Adequacy}

We are now going to prove the soundness of our type system by showing that it
is compatible with our realizability model. This property is specified by the
following theorem which is traditionally called the adequacy lemma.
\begin{definition}
Let $\Gamma$ be a (valid) context. We say that the substitution $\sigma$
realizes $\Gamma$ if:
\begin{itemize}
\item for every $x : A$ in $\Gamma$ we have
  $\sigma(x) \in \llbracket A \rrbracket_\sigma$,
\item for every $\alpha : \lnot A$ in $\Gamma$ we have
  $\sigma(\alpha) \in \llbracket A \rrbracket_\sigma^\bot$,
\item for every $a : Term$ in $\Gamma$ we have
  $\sigma(a) \in \Lambda$,
\item for every $X_n : Pred_n$ in $\Gamma$ we have
  $\sigma(X_n) \in \Lambda^n \to \Lambda_v/\!\!\equiv$,
\item for every $t \equiv u$ in $\Gamma$ we have
  $t\sigma \equiv u\sigma$ and
\item for every $t \not\equiv u$ in $\Gamma$ we have
  $t\sigma \not\equiv u\sigma$.
\end{itemize}
\end{definition}

\begin{theorem}{\emph{(Adequacy.)}}\label{adequacy}
Let $\Gamma$ be a (valid) context, $A$ be a formula such that
$FV(A) \subseteq dom(\Gamma)$ and $\sigma$ be a substitution realizing
$\Gamma$.
\begin{itemize}
\item If $\Gamma \vvdash v : A$ then
  $v\sigma \in \llbracket A \rrbracket_\sigma$,
\item if $\Gamma \vdash t  : A$ then
  $t\sigma \in \llbracket A \rrbracket_\sigma^{\bot\bot}$.
\end{itemize}
\end{theorem}
\input{adequacy}

\begin{remark}
For the sake of simplicity we fixed a pole $\dbot$ at the beginning of the
current section. However, many of the properties presented here (including
the adequacy lemma) remain valid with similar poles. We will make use of
this fact in the proof of the following theorem.
\end{remark}

\begin{theorem}{\emph{(Safety.)}}\label{safety}
Let $\Gamma$ be a context, $A$ be a formula such that
$FV(A) \subseteq dom(\Gamma)$ and $\sigma$ be a substitution realizing
$\Gamma$. If $t$ is a term such that $\Gamma \vdash t : A$ and if $A[\sigma]$
is pure (i.e. it does not contain any $\_ \Rightarrow \_$), then for every
stack $\pi \in \llbracket A \rrbracket_\sigma^\bot$ there is a value
$v \in \llbracket A \rrbracket_\sigma$ and $\alpha \in \mathcal{V}_\mu$ such
that ${t\sigma \ast \pi} \reds {v \ast \alpha}$.
\end{theorem}
\begin{proof}
We do a proof by realizability using the following pole.
$$\dbot_A = \{ p \in \Lambda\times\Pi\;|\;p \reds v \ast \alpha \;\land\; v
\in \llbracket A \rrbracket_\sigma\}$$
It is well-defined as $A$ is pure and hence $\llbracket A \rrbracket_\sigma$
does not depend on the pole. Using the adequacy lemma (theorem \ref{adequacy})
with $\dbot_A$ we obtain $t\sigma \in \llbracket A \rrbracket_\sigma^{\bot\bot}$. Hence for
every stack $\pi \in \llbracket A \rrbracket_\sigma^\bot$ we have
${t\sigma \ast \pi} \in \dbot_A$. We can then conclude using the definition
of the pole $\dbot_A$.
\end{proof}
\begin{remark}
It is easy to see that if $A[\sigma]$ is closed and pure then
$v \in \llbracket A \rrbracket_\sigma$ implies that $\bullet \vdash v : A$.
\end{remark}

\begin{theorem}{\emph{(Consistency.)}}\label{consistency}
There is no $t$ such that $\bullet \vdash t : \bot$.
\end{theorem}
\begin{proof}
Let us suppose that $\bullet \vdash t : \bot$. Using adequacy (theorem
\ref{adequacy}
) we obtain that $t \in \llbracket \bot \rrbracket_\sigma^{\bot\bot}$. Since
$\llbracket \bot \rrbracket_\sigma = \emptyset$ we know that
$\llbracket \bot \rrbracket_\sigma^\bot = \Pi$ by definition. Now using
theorem \ref{poleconsist} we obtain
$\llbracket \bot \rrbracket_\sigma^{\bot\bot} = \emptyset$. This is a
contradiction.
\end{proof}

%% file: rules.tex
\begin{figure}
\centering

\begin{mproof}
\AxiomC{}
\RightLabel{$\text{ax}$}
\UnaryInfC{$\Gamma, x : A \vvdash x : A$}
\end{mproof}
\hfill
\begin{mproof}
\AxiomC{$\Gamma \vvdash v : A$}
\RightLabel{$\uparrow$}
\UnaryInfC{$\Gamma \vdash v : A$}
\end{mproof}
\hfill
\begin{mproof}
\AxiomC{$\Gamma \vdash v : A$}
\RightLabel{$\downarrow$}
\UnaryInfC{$\Gamma \vvdash v : A$}
\end{mproof}
\\[1.5em]

\begin{mproof}
\AxiomC{$\Gamma \vdash t : A \Rightarrow B$}
\AxiomC{$\Gamma \vdash u : A$}
\RightLabel{$\Rightarrow_e$}
\BinaryInfC{$\Gamma \vdash t\; u : B$}
\end{mproof}
\hfill
\begin{mproof}
\AxiomC{$\Gamma, x : A \vdash t : B$}
\RightLabel{$\Rightarrow_i$}
\UnaryInfC{$\Gamma \vvdash \lambda x\;t : A \Rightarrow B$}
\end{mproof}
\\[1.5em]

\begin{mproof}
\AxiomC{$\Gamma, \alpha : \lnot A \vdash t : A$}
\RightLabel{$\mu$}
\UnaryInfC{$\Gamma \vdash \mu \alpha\;t : A$}
\end{mproof}
\hfill
\begin{mproof}
\AxiomC{$\Gamma, \alpha : \lnot A \vdash t : A$}
\RightLabel{$\ast$}
\UnaryInfC{$\Gamma, \alpha : \lnot A \vdash t \ast \alpha : B$}
\end{mproof}
\\[1.5em]

\begin{mproof}
\AxiomC{$\Gamma \vvdash v : A$}
\RightLabel{$\in_i$}
\UnaryInfC{$\Gamma \vvdash v : v \in A$}
\end{mproof}
\hfill
\begin{mproof}
\AxiomC{$\Gamma, x : A, x \equiv u \vdash t : A$}
\RightLabel{$\in_e$}
\UnaryInfC{$\Gamma, x : u \in A \vdash t : A$}
\end{mproof}
\\[1.5em]

\begin{mproof}
\AxiomC{$\Gamma, u_1 \equiv u_2 \vdash t : A$}
\RightLabel{$\restriction_i$}
\UnaryInfC{$\Gamma, u_1 \equiv u_2 \vdash t : A \restriction u_1 \equiv u_2$}
\end{mproof}
\hfill
\begin{mproof}
\AxiomC{$\Gamma, x : A, u_1 \equiv u_2 \vdash t : B$}
\RightLabel{$\restriction_e$}
\UnaryInfC{$\Gamma, x : A \restriction u_1 \equiv u_2 \vdash t : B$}
\end{mproof}
\\[1.5em]

\begin{mproof}
\AxiomC{$\Gamma \vvdash v : A$}
\AxiomC{$a \not\in FV(\Gamma)$}
\RightLabel{$\forall_i$}
\BinaryInfC{$\Gamma \vvdash v : \forall a\;A$}
\end{mproof}
\hfill
\begin{mproof}
\AxiomC{$\Gamma \vdash t : \forall a\;A$}
\RightLabel{$\forall_e$}
\UnaryInfC{$\Gamma \vdash t : A[a := u]$}
\end{mproof}
\\[1.5em]

\begin{mproof}
\AxiomC{$\Gamma, y : A \vdash t : B$}
\AxiomC{$a \not\in FV(\Gamma, B) \cup TV(t)$}
\RightLabel{$\exists_e$}
\BinaryInfC{$\Gamma, y : \exists a\;A \vdash t : B$}
\end{mproof}
\hfill
\begin{mproof}
\AxiomC{$\Gamma \vdash t : A[a := u]$}
\RightLabel{$\exists_i$}
\UnaryInfC{$\Gamma \vdash t : \exists a\;A$}
\end{mproof}
\\[1.5em]

\begin{mproof}
\AxiomC{$\Gamma \vvdash v : A$}
\AxiomC{$X_n \not\in FV(\Gamma)$}
\RightLabel{$\forall_I$}
\BinaryInfC{$\Gamma \vvdash v : \forall X_n\;A$}
\end{mproof}
\hfill
\begin{mproof}
\AxiomC{$\Gamma \vdash t : \forall X_n\;A$}
\RightLabel{$\forall_E$}
\UnaryInfC{$\Gamma \vdash t : A[X_n := P]$}
\end{mproof}
\\[1.5em]

\begin{mproof}
\AxiomC{$\Gamma, x : A \vdash t : B$}
\AxiomC{$X_n \not\in FV(\Gamma, B)$}
\RightLabel{$\exists_E$}
\BinaryInfC{$\Gamma, x : \exists X_n\;A \vdash t : B$}
\end{mproof}
\hfill
\begin{mproof}
\AxiomC{$\Gamma \vdash t : A[X_n := P]$}
\RightLabel{$\exists_I$}
\UnaryInfC{$\Gamma \vdash t : \exists X_n\;A$}
\end{mproof}
\\[1.5em]

\begin{mproof}
\AxiomC{$[\Gamma \vvdash v_i : A_i]_{1 \leq i \leq n}$}
\RightLabel{$\times_i$}
\UnaryInfC{$\Gamma \vvdash \{l_i = v_i\}_{i=1}^n : \{l_i : A_i\}_{1 \leq i \leq n}$}
\end{mproof}
\hfill
\begin{mproof}
\AxiomC{$\Gamma \vvdash v : \{l_i : A_i\}_{1 \leq i \leq n}$}
\RightLabel{$\times_e$}
\UnaryInfC{$\Gamma \vdash v.l_i : A_i$}
\end{mproof}
\\[1.5em]

\begin{mproof}
\AxiomC{$\Gamma \vvdash v : A_i$}
\RightLabel{$+_i$}
\UnaryInfC{$\Gamma \vvdash C_i[v] : [C_i : A_i]_{1 \leq i \leq n}$}
\end{mproof}
\\[1.5em]

\begin{mproof}
\AxiomC{$\Gamma \vvdash v : [C_i : A_i]_{1 \leq i \leq n}$}
\AxiomC{$[\Gamma, x:A_i, C_i[x] \equiv v \vdash t_i : B]_{1 \leq i \leq n}$}
\RightLabel{$+_e$}
\BinaryInfC{$\Gamma \vdash case_v\;[C_i[x] \to t_i]_{1 \leq i \leq n} : B$}
\end{mproof}
\\[1.5em]
 
\begin{mproof}
\AxiomC{$\Gamma, w_1 \equiv w_2 \vdash t[x := w_1] : A$}
\RightLabel{$\equiv_{v,l}$}
\UnaryInfC{$\Gamma, w_1 \equiv w_2 \vdash t[x := w_2] : A$}
\end{mproof}
\hfill
\begin{mproof}
\AxiomC{$\Gamma, t_1 \equiv t_2 \vdash t[a := t_1] : A$}
\RightLabel{$\equiv_{t,l}$}
\UnaryInfC{$\Gamma, t_1 \equiv t_2 \vdash t[a := t_2] : A$}
\end{mproof}
\\[1.5em]
 
\begin{mproof}
\AxiomC{$\Gamma, w_1 \equiv w_2 \vdash t : A[x := w_1]$}
\RightLabel{$\equiv_{v,r}$}
\UnaryInfC{$\Gamma, w_1 \equiv w_2 \vdash t : A[x := w_2]$}
\end{mproof}
\hfill
\begin{mproof}
\AxiomC{$\Gamma, t_1 \equiv t_2 \vdash t : A[a := t_1]$}
\RightLabel{$\equiv_{t,r}$}
\UnaryInfC{$\Gamma, t_1 \equiv t_2 \vdash t : A[a := t_2]$}
\end{mproof}

\caption{Second-order type system.}
\label{pml2rules}
\end{figure}

%% file: adequacy.tex
\begin{proof}
We proceed by induction on the derivation of the judgement $\Gamma\vvdash v:A$
(resp. $\Gamma\vdash t:A$) and we reason by case on the last rule used.

\smallskip\noindent($\text{ax}$)
By hypothesis $\sigma$ realizes $\Gamma, x : A$ from which we directly
obtain $x\sigma \in \vsem{A}_\sigma$.

\smallskip\noindent($\uparrow$) and ($\downarrow$)
are direct consequences of lemma \ref{orthosimple} and theorem \ref{biortho}
respectively.

\smallskip\noindent($\Rightarrow_e$)
We need to prove that $t\sigma\;u\sigma \in \tsem{B}_\sigma$, hence we take
$\pi \in \ssem{B}_\sigma$ and show $t\sigma\;u\sigma \ast \pi \in \dbot$.
Since $\dbot$ is saturated, we can take a reduction step and show
$u\sigma \ast [t\sigma]\pi \in \dbot$. By induction hypothesis $u\sigma \in
\tsem{A}_\sigma$ so we only have to show $[t\sigma]\pi \in \ssem{A}_\sigma$.
To do so we take $v \in \vsem{A}_\sigma$ and show $v \ast [t\sigma]\pi \in
\dbot$. Here we can again take a reduction step and show $t\sigma \ast v.\pi
\in \dbot$. By induction hypothesis we have $t\sigma \in
\tsem{A \Rightarrow B}_\sigma$, hence it is enough to show
$v.\pi \in \ssem{A \Rightarrow B}_\sigma$. We now take a value
$\lambda x\;t_x \in \vsem{A \Rightarrow B}_\sigma$ and show that
$\lambda x\;t_x \ast v.\pi \in \dbot$. We then apply again a reduction step
and show $t_x[x := v] \ast \pi \in \dbot$. Since $\pi \in \ssem{B}_\sigma$
we only need to show $t_x[x := v] \in \tsem{B}_\sigma$ which is true by
definition of $\vsem{A \Rightarrow B}_\sigma$.

\smallskip\noindent($\Rightarrow_i$)
We need to show $\lambda x\;t\sigma \in \vsem{A \Rightarrow B}_\sigma$ so
we take $v \in \vsem{A}_\sigma$ and show $t\sigma[x \!:=\! v] \in
\tsem{B}_\sigma$. Since $\sigma[x := v]$ realizes $\Gamma, x:A$ we can
conclude using the induction hypothesis.

\smallskip\noindent($\mu$)
We need to show that $\mu \alpha\;t\sigma \in \tsem{A}_\sigma$ hence we take
$\pi \in \ssem{A}_\sigma$ and show $\mu \alpha\;t\sigma \ast \pi \in \dbot$.
Since $\dbot$ is saturated, it is enough to show $t\sigma[\alpha := \pi] \ast
\pi \in \dbot$. As $\sigma[\alpha := \pi]$ realizes $\Gamma, \alpha:\lnot A$
we conclude by induction hypothesis.

\smallskip\noindent($\ast$)
We need to show $t\sigma \ast \alpha\sigma \in \tsem{B}_\sigma$, hence we take
$\pi \in \ssem{B}_\sigma$ and show that $(t\sigma \ast \alpha\sigma) \ast \pi
\in \dbot$. Since $\dbot$ is saturated, we can take a reduction step and show
$t\sigma \ast \alpha\sigma \in \dbot$. By induction hypothesis $t\sigma \in
\tsem{A}_\sigma$ hence it is enough to show $\alpha\sigma \in \ssem{A}_\sigma$
which is true by hypothesis.

\smallskip\noindent($\in_i$)
We need to show $v\sigma \in \vsem{v \in A}_\sigma$. We have $v\sigma \in
\vsem{A}_\sigma$ by induction hypothesis, and $v\sigma \equiv v\sigma$ by
reflexivity of $(\equiv)$.

\smallskip\noindent($\in_e$)
By hypothesis we know that $\sigma$ realizes $\Gamma, x : u \in A$. To be able
to conclude using the induction hypothesis, we need to show that $\sigma$
realizes $\Gamma, x : A, x \equiv u$. Since we have $\sigma(x) \in
\vsem{u \in A}_\sigma$, we obtain that $x\sigma \in \vsem{A}_\sigma$ and
$x\sigma \equiv u\sigma$ by definition of $\vsem{u \in A}_\sigma$.

\smallskip\noindent($\restriction_i$)
We need to show $t\sigma \in \tsem{A \restriction u_1 \equiv u_2}_\sigma$.
By hypothesis $u_1\sigma \equiv u_2\sigma$, hence $\vsem{A \restriction u_1
\equiv u_2}_\sigma = \vsem{A}_\sigma$. Consequently, it is enough to show
that $t\sigma \in \tsem{A}_\sigma$, which is exactly the induction hypothesis.

\smallskip\noindent($\restriction_e$)
By hypothesis we know that $\sigma$ realizes $\Gamma, x : A \restriction u_1
\equiv u_2$. To be able to use the induction hypothesis, we need to show that
$\sigma$ realizes $\Gamma, x : A, u_1 \equiv u_2$. Since we have $\sigma(x)
\in \vsem{A \restriction u_1 \equiv u_2}_\sigma$, we obtain that $x\sigma \in
\vsem{A}_\sigma$ and that $u_1\sigma \equiv u_2\sigma$ by definition of
$\vsem{A \restriction u_1 \equiv u_2}_\sigma$.

\smallskip\noindent($\forall_i$)
We need to show that $v\sigma \in \vsem{\forall a\; A}_\sigma = \bigcap_{t
\in \Lambda} \vsem{A}_{\sigma[a := t]}$ so we take $t \in \Lambda$ and show
$v\sigma \in \llbracket A \rrbracket_{\sigma[a := t]}$. This is true by
induction hypothesis since $a \not\in FV(\Gamma)$ and hence $\sigma[a:=t]$
realizes $\Gamma$.

\smallskip\noindent($\forall_e$)
We need to show $t\sigma \in \tsem{A[a := u]}_\sigma = \tsem{A}_{\sigma[a :=
u\sigma]}$ for some $u \in \Lambda$. By induction hypothesis we know $t\sigma
\in \tsem{\forall a\;A}_\sigma$, hence we only need to show that
$\tsem{\forall a\; A}_\sigma \subseteq \tsem{A}_{\sigma[a := u\sigma]}$. By
definition we have $\vsem{\forall a\; A}_\sigma \subseteq \vsem{A}_{\sigma[a
:= u\sigma]}$ so we can conclude using lemma \ref{orthoblabla}.

\smallskip\noindent($\exists_e$)
By hypothesis we know that $\sigma$ realizes $\Gamma, x : \exists a\;A$. In
particular, we know that $\sigma(x) \in \vsem{\exists a\;A}_\sigma$, which
means that there is a term $u \in \Lambda^\ast$ such that $\sigma(x) \in
\vsem{A}_{\sigma[a := u]}$. Since $a \notin FV(\Gamma)$, we obtain that the
substitution $\sigma[a := u]$ realizes the context $\Gamma, x : A$. Using the
induction hypothesis, we finally get $t\sigma = t\sigma[a := u] \in
\tsem{B}_{\sigma[a := u]} = \tsem{B}_\sigma$ since $a \notin TV(t)$ and $a
\notin FV(B)$.

\smallskip\noindent($\exists_i$)
The proof for this rule is similar to the one for \emph{($\forall_e$)}. We
need to show that $\tsem{A[a := u]}_\sigma = \tsem{A}_{\sigma[a := u\sigma]}
\subseteq \tsem{\exists a\;A}_\sigma$. This follows from lemma
\ref{orthoblabla} since $\vsem{A}_{\sigma[a := u\sigma]} \subseteq
\vsem{\exists a\;A}_\sigma$ by definition.

\smallskip\noindent($\forall_I$), $(\forall_E)$, $(\exists_E)$ and $(\exists_I)$
are similar to similar to ($\forall_i$), ($\forall_e$), ($\exists_e$) and
($\exists_i$).

\smallskip\noindent($\times_i$)
We need to show that $\{l_i = v_i\sigma\}_{i \in I} \in
\vsem{\{l_i : A_i\}_{i \in I}}_\sigma$. By definition we need to show that
for all $i \in I$ we have $v_i\sigma \in \vsem{A_i}_\sigma$. This is
immediate by induction hypothesis.

\smallskip\noindent($\times_e$)
We need to show that $v\sigma.l_i \in \tsem{A_i}_\sigma$ for some $i \in I$.
By induction hypothesis we have $v\sigma \in
\vsem{\{l_i : A_i\}_{i \in I}}_\sigma$ and hence $v$ has the form
$\{l_i = v_i\}_{i \in I}$ with $v_i\sigma \in \vsem{A_i}_\sigma$. Let us now
take $\pi \in \ssem{A_i}_\sigma$ and show that
$\{l_i = v_i\sigma\}_{i \in I}.l_i \ast \pi \in \dbot$. Since $\dbot$ is
saturated, it is enough to show $v_i\sigma \ast \pi \in \dbot$. This is true
since $v_i\sigma \in \vsem{A_i}_\sigma$ and $\pi \in \ssem{A_i}_\sigma$.

\smallskip\noindent($+_i$)
We need to show $C_i[v\sigma] \in \vsem{[C_i : A_i]_{i \in I}}_\sigma$ for
some $i \in I$. By induction hypothesis $v\sigma \in \vsem{A_i}_\sigma$ and
hence we can conclude by definition of $\vsem{[C_i : A_i]_{i \in I}}_\sigma$.

\smallskip\noindent($+_e$)
We need to show $case_{v\sigma}\;[C_i[x] \to t_i\sigma]_{i \in I} \in
\tsem{B}_\sigma$. By induction hypothesis $v\sigma \in \vsem{[C_i of A_i]_{i
\in I}}_\sigma$ which means that there is $i \in I$ and $w \in
\vsem{A_i}_\sigma$ such that $v\sigma = C_i[w]$. We take $\pi \in
\ssem{B}_\sigma$ and show $\tcase{C_i[w]}{C_i[x] \to t_i\sigma}_{i \in I}
\ast \pi \in \dbot$. Since $\dbot$ is saturated, it is enough to show
$t_i\sigma[x := w] \ast \pi \in \dbot$. It remains to show that
$t_i\sigma[x := w] \in \tsem{B}_\sigma$. To be able to conclude using
the induction hypothesis we need to show that $\sigma[x := w]$ realizes
$\Gamma, x : A_i, C_i[x] \equiv v$. This is true since $\sigma$ realizes
$\Gamma$, $w \in \vsem{A_i}_\sigma$ and $C_i[w] \equiv v\sigma$ by reflexivity.

\smallskip\noindent($\equiv_{v,l}$)
We need to show $t[x := w_1]\sigma = t\sigma[x := w_1\sigma] \in
\vsem{A}_\sigma$. By hypothesis we know that $w_1\sigma \equiv w_2\sigma$
from which we can deduce $t\sigma[x := w_1\sigma] \equiv t\sigma[x :=
w_2\sigma]$ by extensionality (theorem \ref{extval}). Since $\vsem{A}_\sigma$
is closed under $(\equiv)$ we can conclude using the induction hypothesis.

\smallskip\noindent($\equiv_{t,l}$), ($\equiv_{v,r}$) and ($\equiv_{t,r}$)
are similar to ($\equiv_{v,l}$), using extensionality (theorem \ref{extval}
and theorem \ref{extterm}).
\end{proof}

%% file: equiv.tex
The type system given in figure \ref{pml2rules} does
not provide any way of discharging an equivalence from the context. As a
consequence the truth of an equivalence cannot be used. Furthermore, an
equational contradiction in the context cannot be used to derive falsehood.
To address these two problems, we will rely on a partial decision procedure
for the equivalence of terms. Such a procedure can be easily implemented
using an algorithm similar to Knuth-Bendix, provided that we are able to
extract a set of equational axioms from the definition of $(\equiv)$. In
particular, we will use the following lemma to show that several reduction
rules are contained in $(\equiv)$.

\begin{lemma}\label{equivlem}
Let $t$ and $u$ be terms. If for every stack $\pi \in \Pi$ there is
$p \in \Lambda\times\Pi$ such that $t \ast \pi \succ^* p$ and
$u \ast \pi \succ^* p$ then $t \equiv u$.
\end{lemma}
\begin{proof}
Since $(\succ) \subseteq (\red_i)$ for every $i \in \mathbb{N}$, we can
deduce that $t \ast \pi \red_i^* p$ and $u \ast \pi \red_i^* p$ for every
$i \in \mathbb{N}$. Using lemma \ref{redcompatall}
we can deduce that for every substitution $\sigma$ we have
$t\sigma \ast \pi \red_i^* p\sigma$ and $u\sigma \ast \pi \red_i^* p\sigma$
for all $i \in \mathbb{N}$. Consequently we obtain $t \equiv u$.
\end{proof}
The equivalence relation contains call-by-value $\beta$-reduction,
projection on records and case analysis on variants.
\begin{theorem}
For every $x \in \mathcal{V}_\lambda$, $t \in \Lambda$ and $v \in \Lambda_v$
we have $(\lambda x\;t) v \equiv t[x := v]$.
\end{theorem}
\begin{proof}
  Immediate using lemma \ref{equivlem}.
\end{proof}
\begin{theorem}
For all $k$ such that $1 \leq k \leq n$ we have the following equivalences.
$$(\lambda x\;t) v \equiv t[x := v] \hspace{4em}
\tcase{C_k[v]}{C_i[x_i] \to t_i}_{1 \leq i \leq n} \equiv t_k[x_k := v]$$
\end{theorem}
\begin{proof}
  Immediate using lemma \ref{equivlem}.
\end{proof}

To observe contradictions, we also need to derive some inequivalences on
values. For instance, we would like to deduce a contradiction if two values
with a different head constructor are assumed to be equivalent.
\begin{theorem}
Let $C$, $D \in \mathcal{C}$ be constructors, and $v$, $w \in \Lambda_v$
be values. If $C \neq D$ then $C[v] \not\equiv D[w]$.
\end{theorem}
\begin{proof}
We take $\pi = [\lambda x\; \tcase{x}{C[y] \to y\;|\;D[y] \to \Omega}] \alpha$
where $\Omega$ is an arbitrary diverging term. We then obtain
$C[v] \ast \pi \converge_0$ and $D[w] \ast \pi \diverge_0$.
\end{proof}
\begin{theorem}
Let $\{l_i = v_i\}_{i \in I}$ and $\{l_j = v_j\}_{j \in J}$ be two records.
If $k$ is a index such that $k \in I$ and $k \notin J$ then we have
$\{l_i = v_i\}_{i \in I} \not\equiv \{l_j = v_j\}_{j \in J}$.
\end{theorem}
\begin{proof}
Immediate using the stack $\pi = [\lambda x\; x.l_k] \alpha$.
\end{proof}
\begin{theorem}
For every $x \in \mathcal{V}_\lambda$, $v \in \Lambda_v$, $t \in \Lambda$,
$C \in \mathcal{C}$ and for every record $\{l_i = v_i\}_{i \in I}$ we have
the following inequivalences.
$$
\lambda x\; t \not\equiv C[v]
\quad\quad\quad
\lambda x\; t \not\equiv \{l_i = v_i\}_{i \in I}
\quad\quad\quad
C[v] \not\equiv \{l_i = v_i\}_{i \in I}
$$
\end{theorem}
\begin{proof}
The proof is mostly similar to the proofs of the previous two theorems.
However, there is a subtlety with the second inequivalence. If for every
value $v$ the term $t[x:=v]$ diverges, then we do not have
$\lambda x\;t \not\equiv \{\}$. Indeed, there is no evaluation context (or
stack) that is able to distinguish the empty record $\{\}$ and a diverging
function. To solve this problem, we can extend the language with a new kind
of term $\text{unit}_v$ and extend the relation $(\succ)$ with the
following rule.
$$\text{unit}_{\{\}} \ast \pi \quad\succ\quad \{\} \ast \pi$$
The process $\text{unit}_v \ast \pi$ is stuck for every value $v \neq \{\}$.
The proof can the be completed using the stack
$\pi = [\lambda x\;\text{unit}_x] \alpha$.
\end{proof}

The previous five theorems together with the extensionality of $(\equiv)$ and
its properties as an equivalence relation can be used to implement a partial
decision procedure for equivalence. We will incorporate this procedure into
the typing rules by introducing a new form of judgment.

\begin{definition}
An equational context $\mathcal{E}$ is a list of hypothetical equivalences
and inequivalences. Equational contexts are built using the following grammar.
$$
\mathcal{E} \;:=\;
  \bullet
  \;\;|\;\;
  \mathcal{E}, t \equiv u
  \;\;|\;\;
  \mathcal{E}, t \not\equiv u
$$
Given a context $\Gamma$, we denote $\mathcal{E}_\Gamma$ its restriction to
an equational context.
\end{definition}
\begin{definition}
Let $\mathcal{E}$ be an equational context. The judgement
$\mathcal{E} \vdash \bot$ is valid if and only if the partial decision
procedure is able to derive a contradiction in $\mathcal{E}$. We will write
$\mathcal{E} \vdash t \equiv u$ for $\mathcal{E}, t \not\equiv u \vdash \bot$
and $\mathcal{E} \vdash t \not\equiv u$ for
$\mathcal{E}, t \equiv u \vdash \bot$
\end{definition}

To discharge equations from the context, the following two typing rules are
added to the system.
\begin{prooftree}
  \AxiomC{$\Gamma, u_1 \equiv u_2 \vdash t : A$}
  \AxiomC{$\mathcal{E}_\Gamma \vdash u_1 \equiv u_2$}
  \RightLabel{$\equiv$}
  \BinaryInfC{$\Gamma \vdash t : A$}
\end{prooftree}
\begin{prooftree}
  \AxiomC{$\Gamma, u_1 \not\equiv u_2 \vdash t : A$}
  \AxiomC{$\mathcal{E}_\Gamma \vdash u_1 \not\equiv u_2$}
  \RightLabel{$\not\equiv$}
  \BinaryInfC{$\Gamma \vdash t : A$}
\end{prooftree}
\smallskip
The soundness of these new rules follows easily since the decision procedure
agrees with the semantical notion of equivalence. The axioms that were given
at the beginning of this section are only used to partially reflect the
semantical equivalence relation in the syntax. This is required if we are to
implement the decision procedure.

Another way to use an equational context is to derive a contradiction
directly. For instance, if we have a context $\Gamma$ such that
$\mathcal{E}_\Gamma$ yields a contradiction, one should be able to finish
the corresponding proof. This is particularly useful when working with
variants and case analysis. For instance, some branches of the case analysis
might not be reachable due to constraints on the matched term. For instance,
we know that in the term
$$\tcase{C[v]}{C[x] \to x \;|\; D[x] \to t}$$
the branch corresponding to the $D$ constructor will never be reached.
Consequently, we can replace $t$ by any term and the computation will still
behave correctly. For this purpose we introduce a special value $\scissors$
on which the abstract machine fails. It can be introduced with the following
typing rule.
\begin{prooftree}
  \AxiomC{$\mathcal{E}_\Gamma \vdash \bot$}
  \RightLabel{$\scissors$}
  \UnaryInfC{$\Gamma \vvdash \scissors : \bot$}
\end{prooftree}
\smallskip
The soundness of this rule is again immediate.

%% file: conclusion.tex
The model presented in the previous sections is intended to be used as the
basis for the design of a proof assistant based on a call-by-value ML language
with control operators. A first prototype (with a different theoretical
foundation) was implemented by Christophe Raffalli \cite{pml}. Based on this
experience, the design of a new version of the language with a clean
theoretical basis can now be undertaken. The core of the system will consist
of three independent components: a type-checker, a termination checker and a
decision procedure for equivalence.

Working with a Curry style language has the disadvantage of making
type-checking undecidable. While most proof systems avoid this problem by
switching to Church style, it is possible to use heuristics making most
Curry style programs that arise in practice directly typable. Christophe
Raffalli implemented such a system \cite{normalizer} and from his experience
it would seem that very little help from the user is required in general.
In particular, if a term is typable then it is possible for the user to
provide hints (e.g. the type of a variable) so that type-checking may
succeed. This can be seen as a kind of completeness.

Proof assistants like Coq \cite{coq} or Agda \cite{agda} both have decidable
type-checking algorithms. However, these systems provide mechanisms for
handling implicit arguments or meta-variables which introduce some
incompleteness. This does not make these systems any less usable in practice.
We conjecture that going even further (i.e. full Curry style) provides a
similar user experience.

To obtain a practical programming language we will need support for recursive
programs. For this purpose we plan on adapting Pierre Hyvernat's termination
checker \cite{sct}. It is based on size change termination and has already
been used in the first prototype implementation.
We will also need to extend our type system with inductive (and coinductive)
types \cite{phdraff,mendler}. They can be introduced in the system using
fixpoints $\mu X\,A$ (and $\nu X\,A$).

%% file: acknowledgment.tex
I would like to particularly thank my research advisor, Christophe Raffalli,
for his guidance and input. I would also like to thank Alexandre Miquel for
suggesting the encoding of dependent products. Thank you also to Pierre
Hyvernat, Tom Hirschowitz, Robert Harper and the anonymous reviewers for
their very helpful comments.